\documentclass[journal]{IEEEtran}
\usepackage{graphicx}
\usepackage{epsfig}
\usepackage{latexsym}
\usepackage{amsfonts}
\usepackage{here}
\usepackage{rawfonts}
\usepackage[utf8]{inputenc}
\usepackage[T1]{fontenc}
\usepackage{calc}
\usepackage{url}
\usepackage{enumerate}
\usepackage{color}
\usepackage[tbtags]{amsmath}
\usepackage{amssymb}
\usepackage{upref}
\usepackage{epic,eepic}
\usepackage{times}
\usepackage{dsfont}
\usepackage{comment}
\usepackage{cite}

\usepackage[font={small}]{caption}
\usepackage[font={small}]{subcaption}

\usepackage{stfloats}
\usepackage{mathtools}

\usepackage{algorithm}
\usepackage{algpseudocode}

\newtheorem{theorem}{\bf{Theorem}}
\newtheorem{definition}{\bf{Definition}}

\newtheorem{lemma}{\bf{Lemma}}
\newtheorem{remark}{Remark}

\usepackage{tikz}
\usetikzlibrary{calc}
\usetikzlibrary{arrows.meta,positioning,fit,calc,decorations.pathreplacing}

\ifCLASSINFOpdf
\else
\fi
\begin{document}

	\title{Out-of-Air Computation: Enabling Structured {Function} Extraction from Wireless Superposition}
	
	\author{{Seyed Mohammad Azimi-Abarghouyi,~\IEEEmembership{Member,~IEEE}}

		\thanks{The author is with the Department of Electrical Engineering, Chalmers University of Technology, Gothenburg, Sweden (Email: azimimo@chalmers.se).}
	}

	\maketitle
	
	\vspace{-15pt}

	\begin{abstract}
	{Over-the-air computation (AirComp) broadly exploits the superposition property of wireless multiple-access channels (MACs) to compute functions of distributed data. Within this broad class, dominant conventional designs are embedding-oriented: they pre-shape transmitted signals or mitigate channel effects so that the received superposition directly realizes the prescribed computation, often requiring the MAC to approximate an ideal computational medium. This paper introduces out-of-air computation (AirCPU) and establishes an extraction-oriented paradigm for AirComp. Built on joint source--channel coding, AirCPU creates a structured wireless superposition from which the receiver extracts the target function.} AirCPU operates directly on continuous-valued device data, avoiding the need for a separate source quantization stage, and employs a multi-layer nested lattice architecture that enables progressive resolution by decomposing each input into hierarchically scaled components, all transmitted over a common bounded digital constellation under a fixed power constraint. We formalize the notion of decoupled resolution, showing that in operating regimes where the decoding error probability is sufficiently small, the impact of channel noise and finite constellation constraints on distortion becomes negligible, and the resulting computation error is primarily determined by the target resolution set by the finest lattice. For fading MACs, we further introduce collective and successive computation mechanisms, in addition to the proposed direct computation, which exploit multiple decoded integer-coefficient functions and side-information functions as structural representations of the wireless superposition to significantly expand the reliable operating regime; in this context, we formulate and characterize the underlying reliability conditions and integer optimization problems, and develop structured low-complexity two-group approximations to address them. Simulations over both Gaussian and fading MACs corroborate the analysis and indicate a reliability transition in which the distortion stabilizes once reliable decoding is achieved. The results further demonstrate notable performance improvements over representative analog and digital AirComp baselines {under the considered settings}, and quantify the benefits of collective and successive computation in expanding the reliable operating regime under fading.
	\end{abstract}

	\vspace{0pt}
	\vspace{-12pt}
	\section{Introduction}
	
	Over-the-air computation (AirComp) \cite{nazer_general, goldenbaum} has recently attracted significant attention as a "joint communication--computation" paradigm that transforms wireless multiple-access channels (MACs) into compute engines. By exploiting the superposition property of the wireless medium, AirComp enables a fusion center to directly recover a desired function of distributed edge signals---such as the sum function---from the simultaneous transmissions of all devices, rather than first decoding each individual signal. This principle has been widely advocated for large-scale sensing, Internet of Things (IoT) data aggregation, and distributed and federated learning (FL), where the underlying tasks are inherently dominated by linear aggregation operations. As a result, AirComp promises substantial reductions in communication latency, bandwidth, and energy consumption compared to conventional "communicate-then-compute" architectures \cite{aircompsurvey, aircompmag, airflmag}.
	
	Most existing works on AirComp rely on purely analog transmission architectures \cite{kaibin_aircomp, kaibin_airfl, ding_airfl, gunduz_airfl, deniz_mis, kaibin_aircomp2, azimi_wafel, huang_neural, gunduz_blind, deniz_inference}. In these schemes, each device maps its local data onto continuous-amplitude baseband waveforms through analog modulation, and the fusion center estimates the desired function by applying appropriate equalization to the superimposed received signal. While this fully analog design is conceptually appealing, its practical performance can depend sensitively on amplitude and phase fidelity, as it relies on freely shaped modulation. Moreover, because the computation is carried out directly in the analog domain, the receiver output becomes sensitive to channel fading, phase misalignment, and noise, which may lead to increased error accumulation and reduced accuracy of the computed function, particularly in low-to-moderate signal-to-noise ratio (SNR) regimes. In this setting, reliability cannot be defined independently of distortion, as the channel always governs the computation error. Also, such architectures may incur relatively high implementation costs, as dedicated analog modulators need to be implemented on participating devices, which can limit scalability as the number of devices grows.
	
	To overcome these limitations, a new paradigm has recently emerged that seeks to realize AirComp using \emph{digital} communication techniques, often referred to as digital AirComp \cite{deniz_onebit, saeed_channelcomp, saeed_sumcomp, saeed_blind, deniz_digital, brinton_digital, kaibin_digital, huang_digital, zey_digital, saeed_vector}. The central idea is to retain the functional advantages of AirComp while leveraging the mature physical-layer infrastructure of modern digital wireless systems. In particular, digital AirComp architectures build upon existing digital modulation chipsets already deployed in communication standards, enabling seamless integration. By appropriately encoding quantized local data into finite-alphabet symbols through channel coding, these schemes inherit some level of robustness and interoperability from conventional digital links, thereby enabling improved reliability and controllable computation error. However, existing schemes generally require continuous-valued inputs to be quantized before digital transmission. In this context, comparisons between analog and digital AirComp should be interpreted with care, as different works adopt different modeling and design assumptions. For example, some studies evaluate analog AirComp on discretized inputs to enable like-for-like comparison with finite-alphabet digital transmissions \cite{saeed_sumcomp, zey_digital}. Moreover, performance gains may stem from additional system-level design degrees of freedom, such as advanced transceiver and resource optimization \cite{huang_digital}, or nonlinear MAP detection \cite{kaibin_digital}. These observations indicate that the relative performance of analog and digital AirComp depends strongly on system design choices and problem formulation, rather than a clear intrinsic superiority of digital over analog modulation for function computation.
	
The existing analog and digital AirComp literature has evolved along two main design paradigms. The dominant line of work assumes the availability of accurate channel state information (CSI) at the transmitters (CSIT), enabling devices to apply power control and channel inversion so that their signals are coherently aligned at the fusion center \cite{kaibin_aircomp, kaibin_airfl, ding_airfl, gunduz_airfl, deniz_mis, kaibin_aircomp2, azimi_wafel, huang_neural, deniz_inference, deniz_onebit, saeed_channelcomp, saeed_sumcomp, deniz_digital, brinton_digital, kaibin_digital, huang_digital, zey_digital, saeed_vector}. Under these assumptions, carefully designed precoding and receive equalization can significantly reduce computation error; however, such gains rely on channel estimation via feedback mechanisms and on adaptive power-control hardware, which can introduce considerable overhead and system complexity, especially in dense, mobile, or energy-constrained scenarios. Furthermore, in time-varying channels, maintaining accurate CSIT requires predicting future channel states, which adds additional uncertainty, and in poor channel conditions, channel inversion may require excessively high transmission power. This reliance on per-device CSIT and adaptive power control may differ from the operating principles of many current wireless systems---and emerging standards---where devices typically transmit with fixed or only loosely controlled power and without access to instantaneous CSIT \cite{book}, making such assumptions potentially difficult to sustain in practical deployments.
	
	A second design class avoids CSIT by relying on channel statistics and massive antenna arrays at the fusion center, exploiting favorable propagation and the law of large numbers to average out channel randomness and interference \cite{gunduz_blind, saeed_blind}. While this blind paradigm reduces transmitter-side complexity, its performance guarantees hold primarily in asymptotic regimes with very large numbers of antennas and idealized propagation models. In practical edge systems, deploying and maintaining such large-scale arrays for computation may introduce additional cost and implementation complexity.
	
As a result, across both paradigms, AirComp relies on specific system assumptions and associated design choices; when these assumptions fail, computation errors may become significant or the scheme may be inapplicable.

	At a more fundamental level, this sensitivity stems from how computation is realized in traditional AirComp. Computation is largely predetermined: with CSIT, devices carefully pre-shape their signals so that the superposition directly yields the target function, while with massive antenna arrays, undesirable interference terms statistically vanish. {In both cases, channel variations are largely compensated or averaged out so that the MAC approaches an ideal computational medium. This approximation logic largely neutralizes the intrinsic dynamics of the wireless channel and limits the range of scenarios in which these conventional designs can be applied in their original form.}

{In this work, we introduce \textit{out-of-air computation} (AirCPU), a new framework within the broader AirComp paradigm for digital function computation over MACs. AirCPU retains wireless superposition as the physical basis of computation but departs from the approximation logic of conventional designs. In particular, while the dominant conventional AirComp designs follow an embedding principle---imposing the target function onto transmitted waveforms prior to propagation---AirCPU adopts an extraction principle, whereby the receiver reconstructs the desired function directly from the wireless superposition by exploiting its induced structured, information-rich nature. Under this framework, devices transmit blindly, without CSIT, while the uncompensated fading MAC actively shapes the structured observations used for computation, without requiring large antenna arrays. The term \textit{out-of-air} therefore refers to this receiver-side realization and the active computational role of the wireless medium; it does not imply that AirCPU lies outside the broad AirComp category. Rather than enforcing a predetermined computation model, AirCPU develops structured transmission and decoding strategies for reliable function recovery. In doing so, it uncovers previously unexplored computational capabilities of the wireless medium and enables new modes of computation and adaptation to diverse network deployments.}
	
Our framework establishes an end-to-end continuous-to-continuous joint source--channel coding architecture that operates directly on continuous-valued local data. {Unlike existing digital AirComp designs, which first quantize the continuous values and then separately map the quantized outputs to finite-alphabet transmission symbols, AirCPU performs both mappings within a single code, without passing an intermediate quantized representation to an independent channel encoder. This integration allows the source representation and channel protection to be designed jointly, rather than being optimized as two isolated processing stages.}
	
	Starting from continuous-valued inputs, AirCPU employs lattice codes with controllable resolution to jointly quantize and modulate data while simultaneously providing robustness against wireless channels and noise through a multi-layer nested structure. This design enables reliable decoding up to a chosen resolution, yielding sufficiently small decoding error probability under appropriate operating conditions at the fusion center. To this end, each input is decomposed into a hierarchy of lattice components, all transmitted over a common finite-power constellation, resulting in a practical transmitter architecture with a peak-energy constraint. A key feature of AirCPU is its progressive resolution property: the computation output can be reliably recovered at increasingly finer resolutions as additional lattice layers are decoded, without requiring any change in transmit power, constellation size, or signaling structure. Each decoded layer refines the recovered function value in a strictly hierarchical manner, enabling graceful accuracy improvement rather than the abrupt error behavior typical of conventional schemes. Moreover, AirCPU enables a decoupling of the target resolution from both the wireless channel and the finite constellation. Once reliable decoding is achieved for a given set of lattice layers, channel noise and constellation constraints contribute to the computation distortion primarily through rare decoding failures, and the resulting error is dominated by the selected target resolution, determined by the finest lattice in the hierarchy. This behavior differs from {existing digital AirComp designs}, where quantization resolution is intrinsically coupled to constellation order, and improving accuracy requires enlarging the constellation or increasing transmit power. 
	
	We then introduce two novel concepts beyond our \textit{direct computation}: \textit{collective computation}, in which multiple integer-coefficient function representations are decoded and combined to recover the target function, and \textit{successive computation}, in which a previously decoded function is reused as side information to improve target-function recovery. These mechanisms interpret decoded integer-coefficient functions as structural representations of the wireless superposition. For each mechanism, we design the corresponding fusion-center receiver architecture.
	
{Across direct, collective, and successive computation, we develop a unified finite-dimensional reliability analysis based on the probability that the effective disturbance leaves the Voronoi region of the decoding lattice. For the Gaussian MAC, we derive exact layer-wise and end-to-end decoding-error probabilities, two-sided non-asymptotic geometric bounds, and an explicit sufficient condition for recovering the layers required by a selected resolution under a prescribed error budget. For fading MACs, we derive the exact conditional decoding-error probability as a shifted-Gaussian Voronoi integral and its channel-conditioned Gaussian-mixture characterization. Because the finite-constellation residual interference is generally non-Gaussian, we introduce an isotropic Gaussian surrogate that links the effective-noise metric to decoding reliability and yields tractable end-to-end sufficient conditions for all three computation mechanisms. Based on the resulting reliability-oriented coefficient-selection criteria, we develop an exact reduced-complexity bounded-global search for collective computation, together with polynomial-complexity two-group approximations for collective and successive computation.}

	Finally, simulations validate the proposed principles. The results show that AirCPU exhibits a distinctive reliability behavior: above a certain SNR threshold, the computation error, measured in terms of mean-squared error (MSE), is governed by the resolution set by the finest lattice second moment. Across both Gaussian and fading MACs, AirCPU demonstrates substantial performance gains over representative analog and digital AirComp baselines, with a controlled increase in the number of transmissions, while the proposed collective and successive mechanisms provide additional gains by significantly enlarging the reliable region and reducing the SNR required to reach the resolution-determined MSE floor.
	
	\vspace{-6pt}
\section{Lattice Preliminaries}

\subsection{{Positioning Relative to Previous Lattice-Based Work}}

{
	Lattice codes have established but distinct roles in source and channel coding. For source coding, dithered lattice quantization provides near rate-distortion-optimal representations of continuous signals \cite{rzamir} and has also been used for model quantization in FL under orthogonal transmission rather than over a MAC \cite{eldar}. For channel coding, compute-and-forward (CF) maps finite-field messages to nested-lattice codewords. A relay scales or equalizes its noisy superposition and applies lattice decoding to recover an integer linear combination of the lattice codewords. Modulo-lattice processing then maps the decoded combination to a finite-field equation, and sufficiently many independent equations enable recovery of the individual messages \cite{nazer,nazermimo,nazer_suc}. CF thus follows a message-recovery architecture without source layering; our previous CF works \cite{azimi_cf1,azimi_cf2} remain within the same setting. Within this architecture, continuous-source representation remains external to the CF channel code, and the decoded equations primarily serve as auxiliary algebraic objects for message recovery rather than as the desired real-valued computation outputs.
	
AirCPU instead exploits the same lattice-closure property within a joint source--channel computation framework, while changing both what is transmitted and the computational role of the decoder output. Each continuous input is quantized to a real lattice point and decomposed into multiple hierarchically scaled, bounded components, thereby jointly designing its source representation and multi-layer channel transmission. At each layer, the fusion center applies equalization followed by lattice quantization to recover a Euclidean integer-coefficient function of these source-derived components. Across the layers, the decoded functions are aggregated to reconstruct the desired real-valued function; they are neither mapped to finite-field equations nor used to recover the individual messages, but instead constitute direct and meaningful structured components of the computation output.
	
	Correspondingly, CF is characterized by computation rate and the probability of correctly recovering a discrete equation. In AirCPU, the finest lattice determines the target resolution, the channel determines which layers can be reliably recovered, and an erroneous lattice decision contributes to the computation error according to its displacement from the intended lattice point. Hence, decoding errors manifest as graded perturbations of the reconstructed function rather than solely as binary equation-recovery failures. The multi-layer structure consequently enables progressive refinement and separates the selected resolution from channel-induced reliability. Thus, the resemblance of some decoding expressions to those of CF follows from the common lattice-closure property, whereas their source and output domains, end-to-end objectives, representation and transceiver architectures, decoded-lattice semantics, performance characterizations, and coefficient-design objectives are fundamentally different.
	
Finally, AirCPU builds on our earlier FedCPU joint source--channel coding framework \cite{fedcpu,fedcpu_conf}, which was specifically developed for FL aggregation. Whereas FedCPU directly decodes a single channel-adapted, learning-oriented integer combination from one lattice layer at the resolution imposed by that layer, the contributions unique to AirCPU are its multi-layer architecture for progressive prescribed-function recovery and its collective and successive computation mechanisms.}

	\vspace{-6pt}
	\subsection{Multi-Layer Nested Lattice Code}
	
	Consider a finest lattice $\Lambda_{1} \subset \mathbb{R}^n$ generated by a full-rank matrix 
	$\mathbf{G}_{1} \in \mathbb{R}^{n \times n}$:
	\begin{align}
		\Lambda_{1} = \big\{ \mathbf{G}_{1} \mathbf{z} : \mathbf{z} \in \mathbb{Z}^n \big\}.
	\end{align}
	This lattice serves as the basis for constructing an $L$-layer nested lattice structure. 
	Each coarser lattice in the hierarchy is obtained by integer scaling of the previous (finer) lattice:
	\begin{align}
		\Lambda_{\ell+1} \triangleq \rho \Lambda_{\ell}, 
		\qquad \rho \in \mathbb{N},\ \rho \ge 2.
	\end{align}
	Equivalently, the corresponding generator matrices satisfy
	\begin{align}
		\mathbf{G}_{\ell+1} = \rho \mathbf{G}_{\ell}, 
		\qquad 
		\Lambda_{\ell} = \big\{ \mathbf{G}_{\ell} \mathbf{z} : \mathbf{z} \in \mathbb{Z}^n \big\}.
	\end{align}
	This construction ensures the nesting property
	\begin{align}
		\Lambda_{L} \subset \Lambda_{L-1} \subset \cdots \subset \Lambda_{1},
	\end{align}
	where $\Lambda_{1}$ is the finest and $\Lambda_{L}$ is the coarsest lattice in the chain. 
	
	For each layer $\ell$, the quantizer associated with $\Lambda_{\ell}$ is defined as
	\begin{align}
		Q_{\Lambda_{\ell}}(\mathbf{x}) 
		\triangleq 
		\arg\min_{\boldsymbol{\lambda} \in \Lambda_{\ell}}
		\|\mathbf{x} - \boldsymbol{\lambda}\|,
	\end{align}
	which maps any $\mathbf{x}\in\mathbb{R}^n$ to its closest lattice point in $\Lambda_{\ell}$.  
	The set of points quantized to the origin forms the \textit{fundamental Voronoi region} of $\Lambda_{\ell}$:
	\begin{align}
		\mathcal{V}_{\ell} 
		\triangleq 
		\big\{ \mathbf{x} \in \mathbb{R}^n : 
		Q_{\Lambda_{\ell}}(\mathbf{x}) = \mathbf{0} \big\}.
	\end{align}
	These regions tile $\mathbb{R}^n$ under lattice shifts.
	
	Since each lattice is obtained by scaling its predecessor, their Voronoi regions are 
	geometrically similar and satisfy
	\begin{align}
		\mathcal{V}_{\ell+1} = \rho \mathcal{V}_{\ell}.
	\end{align}
	Consequently, their volumes obey
	\begin{align}
		\mathrm{vol}(\mathcal{V}_{\ell+1})
		= \rho^{n} \, \mathrm{vol}(\mathcal{V}_{\ell}).
	\end{align}
	Thus, the quantizers $Q_{\Lambda_{\ell}}(\cdot)$ provide increasingly finer approximations as $\ell$ decreases, 
	yielding a natural multi-resolution representation across the $L$ nested layers.
	
	The performance of each individual quantizer is characterized by the per-dimension second moment of its Voronoi region:
	\begin{align}
		\sigma^{2}(\Lambda_{\ell}) 
		= \frac{1}{n \, \mathrm{vol}(\mathcal{V}_{\ell})}
		\int_{\mathcal{V}_{\ell}} \| \mathbf{v} \|^{2} \, d\mathbf{v}.
	\end{align}
	The per-layer second moment quantifies the distortion introduced by the $\ell$-th quantizer and directly 
	determines the resolution contributed by that layer in a hierarchical decomposition. Accordingly, the finest lattice $\Lambda_{1}$ sets the target minimum achievable computation resolution through $\sigma^{2}(\Lambda_{1})$, which is a controllable design parameter of the system.

	A representative realization of such a nested lattice structure is illustrated in Fig.~\ref{fig:hexagonal_layers}.\footnote{{The theoretical framework developed in this paper applies to general lattice codes and does not rely on any specific lattice construction. The design or optimization of particular lattice families is beyond the scope of this work.}}

{
	The proposed nested structure shares with hierarchical modulation the
	use of coarse and refinement levels. A quantitative comparison is provided in \textit{Remark~2} of Section~IV. The
	hierarchical-constellation AirComp scheme in \cite{saeed} targets
	multiple discrete function outputs in one channel use, whereas AirCPU
	progressively refines one continuous-valued aggregate.
}
	
	\begin{figure}[!t]
		\centering
		\begin{tikzpicture}[line cap=round,line join=round,scale=0.140]
			
			\def\s{0.8}                           
			\pgfmathsetmacro{\rt}{sqrt(3)}
			
			\newcommand{\XY}[2]{%
				\pgfmathsetmacro{\X}{1.5*(#1)*\s}%
				\pgfmathsetmacro{\Y}{\rt*((#2)+(#1)/2)*\s}%
			}
			\newcommand{\Hex}[4][]{%
				\begin{scope}[shift={({#2},{#3})}]
					\draw[#1] (0:#4)--(60:#4)--(120:#4)--(180:#4)--(240:#4)--(300:#4)--cycle;
				\end{scope}
			}
			
			\def\lwF{0.45pt}   
			\def\lwM{1.4pt}    
			\def\lwC{2.6pt}    
			
			\foreach \q in {-8,...,8}{
				\foreach \r in {-8,...,8}{
					\XY{\q}{\r}
					\Hex[black,line width=\lwF]{\X}{\Y}{\s}
					\fill[black] (\X,\Y) circle[radius=0.18];
				}
			}
			
			\foreach \q in {-9,-6,-3,0,3,6,9}{
				\foreach \r in {-9,-6,-3,0,3,6,9}{
					\XY{\q}{\r}
					\Hex[blue!70!black,line width=\lwM]{\X}{\Y}{3*\s}
					\fill[blue!70] (\X,\Y) circle[radius=0.30];
				}
			}
			
			\foreach \q in {-9,0,9}{
				\foreach \r in {-9,0,9}{
					\XY{\q}{\r}
					\Hex[orange!80!black,line width=\lwC]{\X}{\Y}{9*\s}
					\fill[orange!85!black] (\X,\Y) circle[radius=0.42];
				}
			}
		\end{tikzpicture}
		\caption{Two-dimensional hexagonal lattice with three nested layers: fine (black), intermediate (blue), and coarse (orange).}
		\label{fig:hexagonal_layers}
				\vspace{-8pt}
	\end{figure}


	\vspace{-10pt}
	\section{Out-of-Air Transmission Scheme}
	Assume that there are \(K\) devices and a fusion center. The fusion center aims to compute the sum function \(\mathbf{u} = \sum_{k=1}^{K} \mathbf{u}_k\), where \(\mathbf{u}_k\) denotes the continuous-valued local input of device \(k\).\footnote{{In this work, we deliberately focus on sum computation, consistent with the majority of AirComp studies \cite{kaibin_aircomp, kaibin_airfl, ding_airfl, gunduz_airfl, deniz_mis, kaibin_aircomp2, azimi_wafel, deniz_inference, huang_neural, gunduz_blind, deniz_onebit, saeed_sumcomp, saeed_blind, deniz_digital, brinton_digital, kaibin_digital, huang_digital, zey_digital}. This provides a clean formulation aligned with the principal aggregation operation in distributed sensing and FL. More generally, a nomographic function admitting the representation $F(u_1,\ldots,u_K)
			=
			\psi\!\left(\sum_{k=1}^{K} g_k(u_k)\right)$
			can be computed by letting device \(k\) apply the local preprocessing \(g_k(\cdot)\), using the unchanged AirCPU core to recover the inner sum \(s=\sum_{k=1}^{K}g_k(u_k)\), and applying the post-processing \(\psi(s)\) at the fusion center. Thus, the same AirCPU sum-computation core supports a broad class of linear and nonlinear functions, including, but not limited to, weighted sums, \(p\)-norms, the geometric mean, and the root mean square (RMS), over their respective domains \cite{airflmag, aircompmag, aircompsurvey}. We demonstrate the latter two extensions in Fig.~\ref{fig:nonlinear_functions}.}} To enable this computation, AirCPU operates as follows.
	
	Each device \(k\) first normalizes its input data \(\mathbf{u}_k\) using a mapping function \(\phi(\cdot)\),\footnote{{For the sum computation considered here, $\phi(\cdot)$ is an invertible linear normalization, so that $\phi^{-1}\!\left(\sum_{k=1}^{K}\phi(\mathbf u_k)\right) =\sum_{k=1}^{K}\mathbf u_k$. Its inverse is Lipschitz continuous, ensuring that reconstruction errors in the transformed domain translate proportionally to the original domain.}} resulting in
	\begin{align}
		\mathbf{v}_k = \phi(\mathbf{u}_k).
	\end{align}
Then, a dither vector $\mathbf{d}_k$, uniformly drawn from the fundamental Voronoi region $\mathcal{V}_1$, is independently generated across devices and known at the fusion center through shared randomness. { In practice, a shared pseudorandom seed, the device identifier, and a synchronized block index allow device $k$ and the fusion center to regenerate the same device-specific dither. Hence, exchanging $K$ length-$n$ dither vectors per block is replaced by one-time seed provisioning and reliable block-index synchronization.} The dithered signal is
	\begin{align}
		\mathbf{v}_k + \mathbf{d}_k.
	\end{align}
	The resulting vector is then quantized to its nearest finest lattice point as
	\begin{align}
		\mathbf{x}_k &= Q_{\Lambda_1}\big(\mathbf{v}_k + \mathbf{d}_k\big).
	\end{align}
	The addition of dither results in the quantization error becoming
	uniform, allowing it to be statistically characterized by the second moment of the lattice. 
	
	\vspace{-5pt}
	\subsection{Layered Transmission of Nested Lattice Components}
The quantized lattice point $\mathbf{x}_k \in \Lambda_1$ belongs to an infinite constellation spanning the entire Euclidean space, making direct transmission over a power-limited channel challenging. Such an infinite constellation, with a small minimum distance between adjacent points to ensure the desired computation resolution, requires high-resolution digital representation. Moreover, the densely packed structure of $\Lambda_1$ makes the transmitted points highly sensitive to channel noise, leading to unreliable decoding performance, similar to the limitations observed in analog AirComp. 
	
	To enable practical transmission under a fixed power constraint and to ensure robustness against channel noise---thereby providing reliability and error-control properties---we propose the following transmission framework.
	
To obtain a finite layered representation, we assume that the quantized lattice point \(\mathbf{x}_k\) lies within the Voronoi region of an outer lattice \(\Lambda_{L+1}\), i.e.,
\begin{align}
	Q_{\Lambda_{L+1}}(\mathbf{x}_k)=\mathbf{0},
\end{align}
or equivalently \(\mathbf{x}_k \in \mathcal{V}_{L+1}\), where \(\Lambda_{L+1} = \rho \Lambda_L\). This can be ensured by appropriately selecting the input normalization and the number of layers so that the source values of interest remain within this region.

Under this assumption, any lattice point \(\mathbf{x}_k \in \Lambda_{1}\cap \mathcal{V}_{L+1}\) can be decomposed into $L$ layers as
\begin{align}
	\mathbf{x}_k = \sum_{\ell=1}^{L} \mathbf{x}_k^{\ell},
\end{align}
where the $\ell$-th component is defined as
\begin{align}
	\mathbf{x}_k^{\ell}
	= Q_{\Lambda_{\ell}}(\mathbf{x}_k) - Q_{\Lambda_{\ell+1}}(\mathbf{x}_k),
	\qquad \ell = 1,\ldots, L.
\end{align}
The decomposition is exact because the terms telescope and \(Q_{\Lambda_{L+1}}(\mathbf{x}_k)=\mathbf{0}\).

For each layer \(\ell\), the values of \(\mathbf{x}_k^\ell\) lie in the set
\begin{align}
	\mathcal{X}_\ell
	\triangleq
	\Lambda_\ell \cap \mathcal{V}_{\ell+1},
	\qquad \ell=1,\ldots,L,
\end{align}
where a deterministic boundary-assignment rule (described below) ensures uniqueness for points on the Voronoi boundary. This follows from the identity \(\mathbf{x}_k^\ell = Q_{\Lambda_\ell}(\mathbf{x}_k) \bmod \Lambda_{\ell+1}\), which ensures that \(\mathbf{x}_k^\ell \in \mathcal{V}_{\ell+1}\) while remaining a point in \(\Lambda_\ell\). Thus, the symbol at layer \(\ell\) belongs to the bounded constellation \(\mathcal{X}_\ell\).

Because the underlying lattices form a scaled nested chain, all sets \(\mathcal{X}_\ell\), for \(\ell = 1,\ldots,L\), share the same geometric structure up to scaling.

		\vspace{-5pt}
	
	\subsection{Boundary Handling and Composite Constellation}
	Lattice points located on the boundary of the Voronoi region $\mathcal{V}_{\ell+1}$ are shared by several adjacent cells.  
	To ensure a unique and consistent mapping at both the encoder and decoder, a deterministic boundary-assignment rule is applied so that each shared point is assigned to exactly one Voronoi region.  
		After this assignment, every layer $\ell$ has a well-defined and identical \emph{pattern} of possible constellation points.
	Thus, each layer transmits one symbol from the same normalized, bounded constellation shape; what differs across layers is only the scaling needed to satisfy the transmit-power constraint.

	To align this common constellation with the power constraint $P$, we begin by identifying the maximum radius of the Voronoi region of the second-layer lattice $\Lambda_2$:
	\begin{align}
		R_2 \triangleq \max_{\mathbf{x} \in \Lambda_1 \cap \mathcal{V}_2} \|\mathbf{x}\|.
	\end{align}
	We then choose a global normalization factor $\alpha$ such that the scaled region of $\Lambda_2$ fits exactly within the power constraint ball of radius $\sqrt{nP}$:
	\begin{align}
		\alpha R_2 = \sqrt{nP}
		\quad\Longrightarrow\quad
		\alpha = \frac{\sqrt{nP}}{R_2}.
	\end{align}
	This normalization ensures that, under the maximum allowable transmit power, the constellation attains the largest possible minimum distance between adjacent points.

	Since the nested lattices satisfy $\Lambda_\ell = \rho^{\ell-1}\Lambda_1$,
	the appropriate scaling for the constellation at layer $\ell$ is given by the global factor $\alpha$ divided by the expansion factor $\rho^{\ell-1}$.
	Accordingly, the encoded transmitted symbol at layer $\ell$ is expressed as
	\begin{align}
		\label{denormalize}
		\mathbf{u}_k^\ell 
		= \frac{\alpha}{\rho^{\ell-1}}\, \mathbf{x}_k^\ell.
	\end{align}
	
	With this choice of scaling,
	\begin{align}
		\|\mathbf{u}_k^\ell\|^2 \le nP,
	\end{align}
	for all layers $\ell$, ensuring that every transmitted point lies within the permissible power region.  Importantly, the scaling is fixed and independent of the channel realization, reflecting a blind transmission architecture without CSIT.
	As a result, all layers use exactly the same normalized constellation shape, while their amplitudes differ only through the layer-dependent factor $1/\rho^{\ell-1}$ imposed by the nested lattice hierarchy. It is also worth noting that lattice-based constellations can, in general, be implemented using standard digital modulation techniques \cite{nazer, nazermimo, azimi_cf1, fedcpu}.

	A representative constellation is shown in Fig.~2.
	
{
	\noindent\textbf{Device-side complexity:}
	Let $C_Q(n)$ denote the cost of one $n$-dimensional
	nearest-neighbor quantization to $\Lambda_1$. Since
	$Q_{\Lambda_\ell}(\mathbf z)
	=\rho^{\ell-1}Q_{\Lambda_1}(\mathbf z/\rho^{\ell-1})$,
	all layers reuse the same base-lattice quantizer. After
	$\mathbf{x}_k$ is computed, only
	$Q_{\Lambda_\ell}(\mathbf{x}_k)$, $\ell=2,\ldots,L$, require
	additional evaluations because
	$Q_{\Lambda_1}(\mathbf{x}_k)=\mathbf{x}_k$ and
	$Q_{\Lambda_{L+1}}(\mathbf{x}_k)=\mathbf 0$. Thus, the layered
	decomposition adds $(L-1)C_Q(n)+\mathcal{O}((L-1)n)$ operations,
	giving a total cost of $LC_Q(n)+\mathcal{O}(Ln)$ per source block.
	The boundary assignment is a fixed tie-breaking step requiring no
	separate search. Hence, for fixed $n$, the device-side complexity is
	linear in $L$ and independent of $K$ and the channel realization.
}
	
	
	\begin{figure}[t]
		\centering
		\begin{tikzpicture}[line cap=round,line join=round,scale=0.4]
			
			\def\s{0.8}
			\pgfmathsetmacro{\rt}{sqrt(3)}
			
			\newcommand{\XY}[2]{%
				\pgfmathsetmacro{\X}{1.5*(#1)*\s}%
				\pgfmathsetmacro{\Y}{\rt*((#2)+(#1)/2)*\s}%
			}
			
			\newcommand{\Hex}[4][]{%
				\begin{scope}[shift={({#2},{#3})}]
					\draw[dashed,#1] (0:#4)--(60:#4)--(120:#4)--(180:#4)--(240:#4)--(300:#4)--cycle;
				\end{scope}
			}
			
			\def\lwF{0.45pt}
			\def\lwM{1.4pt}
			\def\dotR{0.18}
			
			\newcommand{\CircleDotAt}[2]{%
				\fill[black] (#1,#2) circle[radius=\dotR];
			}
			
			\newcommand{\SmallHexDot}[2]{%
				\XY{#1}{#2}%
				\Hex[black,line width=\lwF]{\X}{\Y}{\s}%
				\CircleDotAt{\X}{\Y}%
			}
			
			\Hex[blue!70!black,line width=\lwM]{0}{0}{3*\s}
			
			\begin{scope}
				\clip (0:3*\s)--(60:3*\s)--(120:3*\s)--(180:3*\s)--(240:3*\s)--(300:3*\s)--cycle;
				
				\foreach \q in {-8,...,8}{
					\foreach \r in {-8,...,8}{
						\XY{\q}{\r}
						\Hex[black,line width=\lwF]{\X}{\Y}{\s}
						
						\pgfmathtruncatemacro{\skipdot}{%
							(\q==-1 && \r== 2) ||%
							(\q== 1 && \r== 1) ||%
							(\q==-1 && \r==-1) ||%
							(\q== 1 && \r==-2) ? 1 : 0}
						\ifnum\skipdot=0
						\CircleDotAt{\X}{\Y}
						\fi
					}
				}
			\end{scope}
			
			\SmallHexDot{-2}{1}
			\SmallHexDot{ 2}{-1}
			
			\Hex[blue!70!black,line width=\lwM]{0}{0}{3*\s}
			
			\draw[dashed,green!80!black,line width=1pt] (0,0) circle[radius=3*\s];
			
		\end{tikzpicture}
		
		\caption{Representative hexagonal-lattice constellation used in each transmission layer. Black points denote the constellation elements, and the green circle indicates the power boundary.}
				\vspace{-8pt}
	\end{figure}

			\vspace{-7pt}
	
	\subsection{Joint Source--Channel Functionality}
	This layer-wise nested mapping naturally achieves a \emph{joint source--channel coding} behavior. 
	The quantization step over $\Lambda_1$ serves as the source-coding operation, 
	compressing the continuous-valued vector into a discrete lattice representation, 
	while the sequential lattice transmission acts as a channel-coding mechanism, 
	spreading the information across multiple layers of finite-power constellations. 
	At the receiver, layer-by-layer decoding from $\Lambda_L$ upward progressively reconstructs $\sum_{k=1}^{K}\mathbf{x}_k$ and, consequently, $\sum_{k=1}^{K}\mathbf{v}_k$ with increasing resolution.
	
	This joint design bridges source and channel coding via nested lattices, enabling efficient end-to-end mapping from the source domain to the channel input 
	with minimal distortion under a fixed power constraint. As the constellation is scaled with the transmission power, the scaled Voronoi regions can confine the channel noise, and can lead to sufficiently small decoding error probability, thereby enabling \textit{reliable} operation.

	\section{Out-of-Air Computation over Gaussian MAC}
	
	We now present the computation process of AirCPU  over a Gaussian MAC.  
	At each lattice layer~$\ell$, the $K$ transmitting devices simultaneously send their encoded lattice symbols $\mathbf{u}_k^{\ell},\forall k$.  
	The received signal at the fusion center can therefore be written as
	\begin{align}
		\mathbf{y}^{\ell}
		= \sum_{k=1}^{K} \mathbf{u}_{k}^{\ell}
		+ \mathbf{n}^{\ell},
	\end{align}
	where $\{\mathbf{n}^{\ell}\}_{\ell=1}^{L}$ are independent AWGN vectors, each distributed as $\mathcal{N}(\mathbf{0},\sigma_n^2\mathbf{I})$.  This model is equivalent under ideal conditions in which perfect channel inversion or asymptotically large antenna arrays neutralize wireless channel effects. For reference, most existing digital AirComp schemes are formulated only under Gaussian MAC models of this form \cite{deniz_onebit, saeed_channelcomp, saeed_sumcomp, saeed_blind, kaibin_digital, huang_digital, zey_digital, saeed_vector, saeed}.

	\noindent\textbf{Lattice Decoding:}
	The fusion center applies a scaled lattice quantization to the received signal by mapping it onto the scaled lattice $\alpha\Lambda_1$, and subsequently denormalizes it according to \eqref{denormalize} as
	\begin{align}
		\label{decoding}
		\mathbf{s}^{\ell}
		= \frac{{\rho}^{\ell-1}}{\alpha} Q_{\alpha \Lambda_1}\left(
		{\mathbf{y}^{\ell}}
		\right),
	\end{align}
	which effectively recovers an estimate of the noiseless target function
	$\sum_{k=1}^{K}\mathbf{x}_k^{\ell}$. We quantize onto $\alpha\Lambda_1$ since $\sum_k \mathbf{u}_k^\ell \in \alpha\Lambda_1$ by \eqref{denormalize} and $\Lambda_\ell=\rho^{\ell-1}\Lambda_1$.
	The decoding operation in \eqref{decoding} can be interpreted as
	lattice decoding under an additive disturbance. In the Gaussian MAC
	considered here, this disturbance is exactly $\mathbf n^\ell$.
	
	Decoding is considered reliable when the quantizer correctly
	identifies the intended lattice point, as defined below.
	
{
	\begin{definition}
		Let
		\begin{align}
			\boldsymbol{\lambda}^{\ell}
			\triangleq
			\sum_{k=1}^{K}\mathbf{u}_{k}^{\ell}
			\in \alpha\Lambda_{1}
		\end{align}
		denote the intended lattice point at layer $\ell$, and write
		$\mathbf{y}^{\ell}=\boldsymbol{\lambda}^{\ell}+\mathbf{z}^{\ell}$,
		where $\mathbf{z}^{\ell}$ is the effective additive noise. The
		layer-$\ell$ lattice decoding error probability is
		\begin{align}
			P_{e}^{\ell}
			\triangleq
			\mathbb{P}\!\left(
			Q_{\alpha\Lambda_{1}}(\mathbf{y}^{\ell})
			\neq \boldsymbol{\lambda}^{\ell}
			\right).
		\end{align}
		Decoding at this layer is said to be reliable at level $\epsilon$ if
		$P_{e}^{\ell}\leq\epsilon$.
	\end{definition}
	
	Because nearest-neighbor lattice decoding is translation invariant,
	correct decoding occurs if and only if the effective noise lies in the
	Voronoi region of the decoding lattice. Therefore, for any effective-noise
	density $f_{\mathbf{z}^{\ell}}$, the exact finite-dimensional error
	probability is
	\begin{align}
		\label{eq:exact_pe_general}
		P_{e}^{\ell}
		=
		1-
		\int_{\alpha\mathcal{V}_{1}}
		f_{\mathbf{z}^{\ell}}(\mathbf{z})\,\mathrm{d}\mathbf{z}.
	\end{align}
	In particular, if
	$\mathbf{z}^{\ell}\sim\mathcal{N}(\mathbf{0},v\mathbf{I})$, define
	\begin{align}
		\label{eq:exact_pe}
		\mathcal{E}_{\alpha\Lambda_{1}}(v)
		&\triangleq
		1-
		\frac{1}{(2\pi v)^{n/2}}
		\int_{\alpha\mathcal{V}_{1}}
		\exp\!\left(-\frac{\|\mathbf{z}\|^{2}}{2v}\right)
		\,\mathrm{d}\mathbf{z}.
	\end{align}
Since directly evaluating \eqref{eq:exact_pe} depends on the geometry of
$\mathcal{V}_{1}$, we derive tractable bounds in terms of its inradius and covering radius, defined respectively as \begin{align} r_{\mathrm{in},1} &\triangleq \min_{\mathbf{z}\in\partial\mathcal{V}_{1}}\|\mathbf{z}\|, & r_{\mathrm{cov},1} &\triangleq \max_{\mathbf{z}\in\mathcal{V}_{1}}\|\mathbf{z}\|. \end{align}
	
	\begin{lemma}
		For every $v>0$, the exact error probability in
		\eqref{eq:exact_pe} satisfies the two-sided non-asymptotic bounds
		\begin{align}
			\mathcal{E}_{\alpha\Lambda_{1}}(v)
			&\geq
			\frac{
				\Gamma\!\left(
				\frac{n}{2},
				\frac{\alpha^{2}r_{\mathrm{cov},1}^{2}}{2v}
				\right)}
			{\Gamma\!\left(\frac{n}{2}\right)},
			\label{eq:pe_lower_bound}
			\\
			\mathcal{E}_{\alpha\Lambda_{1}}(v)
			&\leq
			\frac{
				\Gamma\!\left(
				\frac{n}{2},
				\frac{\alpha^{2}r_{\mathrm{in},1}^{2}}{2v}
				\right)}
			{\Gamma\!\left(\frac{n}{2}\right)},
			\label{eq:pe_upper_bound}
		\end{align}
		where $\Gamma(s,x)$ is the upper incomplete gamma function.
	\end{lemma}
	
	\begin{IEEEproof}
		The scaled Voronoi region satisfies
		\begin{align}
			\big\{\mathbf{z}:\|\mathbf{z}\|
			\leq\alpha r_{\mathrm{in},1}\big\}
			\subseteq
			\alpha\mathcal{V}_{1}
			\subseteq
			\big\{\mathbf{z}:\|\mathbf{z}\|
			\leq\alpha r_{\mathrm{cov},1}\big\}.
		\end{align}
		Taking complements gives the corresponding upper and lower bounds
		on the decoding-error event. Since
		$\|\mathbf{z}^{\ell}\|^{2}/v$ follows a chi-square distribution with
		$n$ degrees of freedom, its tail probability is the regularized
		upper incomplete gamma function, which yields
		\eqref{eq:pe_lower_bound} and \eqref{eq:pe_upper_bound}.
	\end{IEEEproof}
	
	Let $\chi^{2}_{n,1-\epsilon}$ denote the $(1-\epsilon)$-quantile of a
	chi-square random variable with $n$ degrees of freedom. Inverting
	\eqref{eq:pe_upper_bound} gives the explicit sufficient condition
	\begin{align}
		v
		\leq
		\sigma_{\mathrm{rel}}^{2}(\epsilon)
		&\triangleq
		\frac{\alpha^{2}r_{\mathrm{in},1}^{2}}
		{\chi^{2}_{n,1-\epsilon}}
	=
		\frac{nP\,r_{\mathrm{in},1}^{2}}
		{R_{2}^{2}\chi^{2}_{n,1-\epsilon}},
		\label{eq:nonasymptotic_reliability}
	\end{align}
	where the second equality follows from
	$\alpha R_{2}=\sqrt{nP}$. Hence,
	\eqref{eq:nonasymptotic_reliability} relates the admissible effective-noise
	variance to $P$, $n$, $\epsilon$, and the lattice ratio
	$r_{\mathrm{in},1}/R_{2}$, with the dependence on $\rho$ entering through
	$R_{2}$.
	
For the Gaussian MAC, each layer-wise error event depends only on the
corresponding AWGN vector $\mathbf n^\ell$. Since these vectors are
independent and identically distributed across layers, the probability
that at least one of the $L$ layer-wise lattice decisions is incorrect is \begin{align} P_{e,\mathrm{tot}} = 1- \left[ 1-\mathcal{E}_{\alpha\Lambda_{1}}(\sigma_n^2) \right]^L. \end{align} For a prescribed end-to-end error target $\epsilon_{\mathrm{tot}}\in(0,1)$, define \begin{align} \epsilon_{\mathrm{lay}} \triangleq 1- (1-\epsilon_{\mathrm{tot}})^{1/L}. \end{align} Then $P_{e,\mathrm{tot}}\leq\epsilon_{\mathrm{tot}}$ is equivalent to \begin{align} \mathcal{E}_{\alpha\Lambda_{1}}(\sigma_n^2) \leq \epsilon_{\mathrm{lay}}. \end{align} Combining this relation with the upper bound in \eqref{eq:pe_upper_bound} gives the sufficient condition \begin{align} \label{reliable} \sigma_n^{2} \leq \sigma_{\mathrm{rel}}^{2} \!\left(\epsilon_{\mathrm{lay}}\right). \end{align} For small $\epsilon_{\mathrm{tot}}$, $\epsilon_{\mathrm{lay}}\simeq\epsilon_{\mathrm{tot}}/L$.
}

	\begin{figure}[h]
		\centering
		\begin{tikzpicture}[line cap=round,line join=round,scale=0.7]
			\def\s{0.8}
			\pgfmathsetmacro{\rt}{sqrt(3)/2}
			
			\draw[black,line width=0.8pt]
			(0:\s)--(60:\s)--(120:\s)--(180:\s)--(240:\s)--(300:\s)--cycle;
			
			\fill[red!70!black!40!white] (0,0) circle[radius=\s*\rt];
			
			\fill[black] (0,0) circle[radius=0.18];
		\end{tikzpicture}
		\caption{{Hexagonal Voronoi region and its inscribed ball. The exact decoding error probability is the Gaussian probability mass outside the hexagonal Voronoi region, whereas replacing the region by the inscribed ball yields the rigorous upper bound in \eqref{eq:pe_upper_bound}.}}
	\end{figure}

{
	\begin{remark}
		Equations~\eqref{eq:exact_pe_general}--\eqref{eq:pe_upper_bound}
		characterize the probability that the decoder leaves the correct
		Voronoi cell. This probability should be distinguished from the
		magnitude of the resulting computation error. Conditional on a
		decoding error, nearest-neighbor decoding returns another lattice
		point, and the resulting distortion is governed by the displacement
		between that point and the intended lattice point. Hence,
		$P_{e,\mathrm{tot}}$ quantifies the incidence of decoding errors,
		while their conditional displacement determines their MSE
		contribution. This preserves the graded-error interpretation of
		AirCPU while providing an explicit reliability guarantee.
	\end{remark}
}

\noindent\textbf{Layer Aggregation:} { The fusion center uses the nearest-neighbor output in
	\eqref{decoding} for every transmitted layer, regardless of whether
	the corresponding lattice decision is correct. To characterize the progressive source-side resolution, we condition on
	correct recovery of the $q\in\{1,\ldots,L\}$ coarsest layers in the order
	$L,L-1,\ldots,1$, and define the partial aggregate
	\begin{align}
		\mathbf s^{(q)}
		\triangleq
		\sum_{\ell=L-q+1}^{L}\mathbf s^{\ell}.
		\label{eq:partial_layer_aggregate}
	\end{align}
	Using the telescoping decomposition in the encoder, this aggregate becomes
	\begin{align}
		\mathbf s^{(q)}
		&=
		\sum_{k=1}^{K}
		\sum_{\ell=L-q+1}^{L}\mathbf x_k^{\ell}
		\nonumber\\
		&=
		\sum_{k=1}^{K}
		Q_{\Lambda_{L-q+1}}(\mathbf x_k),
		\label{eq:partial_reconstruction}
	\end{align}
	where $Q_{\Lambda_{L+1}}(\mathbf x_k)=\mathbf 0$ has been used.
	Hence, $q=1$ produces the coarsest aggregate, while $q=L$ recovers
	the complete lattice aggregate. The corresponding progressive
	computation estimate is
	\begin{align}
		\hat{\mathbf u}^{(q)}
		=
		\phi^{-1}\left(
		\mathbf s^{(q)}-
		\sum_{k=1}^{K}\mathbf d_k
		\right).
		\label{eq:progressive_estimate}
	\end{align}
	For $q=L$, this reduces to the final estimate
	\begin{align}
		\label{final_stage}
		\hat{\mathbf u}
		=
		\hat{\mathbf u}^{(L)}
		=
		\phi^{-1}\left(
		\sum_{\ell=1}^{L}\mathbf s^{\ell}
		-
		\sum_{k=1}^{K}\mathbf d_k
		\right).
	\end{align}
	
	\subsubsection{Quantitative Progressive-Resolution Trade-off}
	\label{sec:progressive_tradeoff}
	The deterministic boundary assignment makes $\mathcal X_\ell$ a
	set of coset representatives of
	$\Lambda_\ell/\Lambda_{\ell+1}$. Therefore,
	\begin{align}
		|\mathcal X_\ell|
		&=
		[\Lambda_\ell:\Lambda_{\ell+1}]
		=
		\frac{\operatorname{vol}(\mathcal V_{\ell+1})}
		{\operatorname{vol}(\mathcal V_\ell)}
		=
		\rho^n.
		\label{eq:layer_cardinality}
	\end{align}
	Since one $n$-dimensional layer symbol occupies $n$ real channel
	dimensions, the per-device uncoded spectral efficiency of every
	AirCPU layer is
	\begin{align}
		\eta_{\mathrm{AirCPU}}
		\triangleq
		\frac{1}{n}\log_2|\mathcal X_\ell|
		=
		\log_2\rho,
		\label{eq:aircpu_spectral_efficiency}
	\end{align}
	measured in bits per real channel dimension. The first $q$ layers in the coarse-to-fine order carry the cumulative
	representation rate
	\begin{align}
		\mathcal R_q
		=
		q\log_2\rho
		\quad
		\text{bits per source dimension}.
		\label{eq:progressive_rate}
	\end{align}

		{
		For this conditional analysis, define the transformed-domain source reconstruction
			\begin{align}
				\widehat{\mathbf v}^{(q)}
				&\triangleq
				\mathbf s^{(q)}-\sum_{k=1}^{K}\mathbf d_k
				\nonumber\\
				&=
				\sum_{k=1}^{K}
				Q_{\Lambda_{L-q+1}}(\mathbf x_k)
				-
				\sum_{k=1}^{K}\mathbf d_k,
			\end{align}
			and its source-side distortion
			\begin{align}
				\mathcal D_q
				\triangleq
				\mathbb E\!\left[
				\left\|
				\widehat{\mathbf v}^{(q)}
				-
				\sum_{k=1}^{K}\mathbf v_k
				\right\|^2
				\right],
			\end{align}
			where the expectation is over the source vectors and independent
			dithers. The unresolved refinement relative to the correct
			full-resolution reconstruction is exactly
			\begin{align}
				\boldsymbol{\Delta}^{(q)}
				&\triangleq
				\widehat{\mathbf v}^{(L)}
				-
				\widehat{\mathbf v}^{(q)}
				=
				\sum_{k=1}^{K}
				\sum_{\ell=1}^{L-q}\mathbf x_k^\ell .
			\end{align}
		}
		Since
		$\mathcal X_\ell=\rho^{\ell-1}\mathcal X_1$ and
		$\max_{\mathbf x\in\mathcal X_1}\|\mathbf x\|=R_2$,
		\begin{align}
			\|\boldsymbol{\Delta}^{(q)}\|
			\leq
			KR_2
			\frac{\rho^{L-q}-1}{\rho-1}.
		\end{align}
		Consequently, Minkowski's inequality gives the following progressive distortion--rate upper bound:
		\begin{align}
			\sqrt{\mathcal D_q}
			\leq
			\sqrt{\mathcal D_L}
			+
			\frac{KR_2}{\rho-1}
			\left[
			2^{\mathcal R_L-\mathcal R_q}-1
			\right],
			\label{eq:progressive_distortion}
		\end{align}
		where, under independent subtractive dithering,
		\begin{align}
			\mathcal D_L
			=
			nK\sigma^2(\Lambda_1).
		\end{align}
		Thus, each additional decoded layer contributes $\log_2\rho$ bits
		per source dimension and monotonically reduces the upper bound on the
		unresolved refinement error. If $L_{\phi^{-1}}$ denotes the Lipschitz constant of $\phi^{-1}$, the
		original-domain computation MSE conditioned on correct recovery of
		these $q$ layers is at most
		$L_{\phi^{-1}}^{2}\mathcal D_q$.
		
	{
		\begin{remark}[Comparison with hierarchical modulation]
			For a representation-level one-shot benchmark based on hierarchical
			modulation \cite{aluini,saeed}, the composite alphabet
			$\Lambda_1\cap\mathcal V_{L+1}$ contains $\rho^{nL}$ points and
			therefore has nominal per-device uncoded spectral efficiency
			$\eta_{\mathrm{HM}}=L\log_2\rho$ bits per real channel dimension.
			AirCPU instead uses a fixed $\rho^n$-point alphabet with
			$\eta_{\mathrm{AirCPU}}=\log_2\rho$ per transmission and conveys the
			same $L\log_2\rho$-bit-per-source-dimension representation over $L$
			channel uses. If the same $q$ nested indices are recovered correctly,
			both representations yield \eqref{eq:partial_reconstruction} and hence
			the same source-side distortion bound in
			\eqref{eq:progressive_distortion}. The one-shot rate advantage,
			however, requires the exponentially larger $\rho^{nL}$-point
			composite alphabet; under the same power constraint, its
			decoding-error probability, and therefore its unconditional
			computation MSE, need not equal that of AirCPU. Thus, hierarchical
			modulation offers an $L$-fold nominal per-use spectral efficiency and
			lower latency, whereas AirCPU trades $L$ channel uses for a fixed-size
			constellation, progressive intermediate estimates, and a direct
			multiuser function-computation rule.
		\end{remark}
	}
	}

{ Let $\mathcal E_{\mathrm{tot}}$ denote the event that at least one of
	the $L$ layer-wise lattice decisions is incorrect. Recalling that
	$P_{e,\mathrm{tot}}=\Pr(\mathcal E_{\mathrm{tot}})$, the law of total
	expectation gives \begin{align} \mathrm{MSE} &= (1-P_{e,\mathrm{tot}}) \mathbb E\!\left[ \|\hat{\mathbf u}-\mathbf u\|^2 \,\middle|\, \mathcal E_{\mathrm{tot}}^{c} \right] \nonumber\\ &\quad+ P_{e,\mathrm{tot}} \mathbb E\!\left[ \|\hat{\mathbf u}-\mathbf u\|^2 \,\middle|\, \mathcal E_{\mathrm{tot}} \right]. \end{align} In the Gaussian MAC, $\mathcal E_{\mathrm{tot}}$ depends only on the AWGN vectors and is independent of the source-quantization errors. Therefore, under correct decoding, 
	\begin{align} \mathbb E\!\left[ \|\hat{\mathbf u}-\mathbf u\|^2 \,\middle|\, \mathcal E_{\mathrm{tot}}^{c} \right] &\leq L_{\phi^{-1}}^{2}\mathcal D_L = L_{\phi^{-1}}^{2}nK\sigma^2(\Lambda_1). 
	\end{align} Hence, \begin{align} \mathrm{MSE} &\leq (1-P_{e,\mathrm{tot}}) L_{\phi^{-1}}^{2}nK\sigma^2(\Lambda_1) \nonumber\\ &\quad+ P_{e,\mathrm{tot}} \mathbb E\!\left[ \|\hat{\mathbf u}-\mathbf u\|^2 \,\middle|\, \mathcal E_{\mathrm{tot}} \right]. \end{align} When $P_{e,\mathrm{tot}}\ll1$ and the decoding-error contribution is negligible relative to the source-quantization term, the MSE is governed by the finest-lattice quantization distortion. For identity normalization, this reduces to $\mathrm{MSE}\approx nK\sigma^2(\Lambda_1)$.}

Under these conditions, the finest lattice determines the distortion
floor, whereas the channel quality and transmit scaling determine
whether the decoder operates sufficiently reliably to attain that
floor. We refer to this separation as \emph{decoupled resolution}.
	
	\section{Out-of-Air Computation over Fading MAC}
	We now present the computation process of AirCPU over a fading MAC, operating directly on the raw, uncompensated, channel-impaired superposition. Consider a set of $K$ single-antenna devices transmitting to an $M$-antenna receiving fusion center during the $\ell$-th communication round, where $M$ is an arbitrary positive integer. The received signal at the $m$-th antenna is given by
	\begin{align}
		\mathbf{y}_m^{\ell}
		= \sum_{k=1}^{K}  h_{k,m}^{\ell}\, \mathbf{u}_{k}^{\ell}
		+ \mathbf{n}_{m}^{\ell},
	\end{align}
	where $h_{k,m}^{\ell}$ is the complex channel coefficient between
	device $k$ and antenna $m$, and $\mathbf n_m^\ell
		\sim
		\mathcal{CN}\!\left(
		\mathbf 0,
		2\sigma_n^2\mathbf I_n
		\right)$
	is the complex AWGN vector, independently across antennas and layers.
	
{ We adopt layer-wise block fading: $h_{k,m}^{\ell}$ remains constant during the transmission slot carrying layer $\ell$, but may change across layers. Hence, if $T_{\mathrm{slot}}$ and $T_{\mathrm{c}}$ denote the slot duration and channel coherence time, respectively, the model requires $T_{\mathrm{slot}}\leq T_{\mathrm{c}}$. Although AirCPU does not require CSIT, the fusion center requires receiver-side channel estimates. Orthogonal uplink training requires at least $K$ pilot symbols per coherence block. If all $L$ layers are transmitted within the same coherence block, these estimates are reused; if the channel changes across layers, training is repeated, requiring at least $LK$ pilot symbols per source block. The present analysis assumes perfect receiver-side CSI, while estimation errors introduce additional equalization mismatch and reduce lattice-decoding reliability.}
	
	Stacking the received signals over all antennas yields the compact complex-valued matrix form
	\begin{align}
		\mathbf{Y}_{\text{c}}^{\ell}
		= \mathbf{H}_{\text{c}}^{\ell} \mathbf{U}^{\ell}
		+ \mathbf{N}_{\text{c}}^{\ell},
	\end{align}
	where
	\begin{align}
		\mathbf{H}_{\text{c}}^{\ell} &=
		\begin{bmatrix}
			h_{1,1}^{\ell} & \cdots & h_{K,1}^{\ell} \\
			\vdots & \ddots & \vdots \\
			h_{1,M}^{\ell} & \cdots & h_{K,M}^{\ell}
		\end{bmatrix}, \\[4pt]
		\mathbf{Y}_{\text{c}}^{\ell} &=
		\begin{bmatrix}
			\mathbf{y}_{1}^{\ell \ \top} \\ \vdots \\ \mathbf{y}_{M}^{\ell \ \top}
		\end{bmatrix}, \quad
		\mathbf{N}_{\text{c}}^{\ell} =
		\begin{bmatrix}
			\mathbf{n}_{1}^{\ell \ \top} \\ \vdots \\ \mathbf{n}_{M}^{\ell \ \top}
		\end{bmatrix}, \quad
		\mathbf{U}^{\ell} =
		\begin{bmatrix}
			\mathbf{u}_{1}^{\ell \ \top} \\ \vdots \\ \mathbf{u}_{K}^{\ell \ \top}
		\end{bmatrix}.
	\end{align}
	
	For convenience, we can represent the above model in its equivalent real-valued form:
	\begin{align}
		\mathbf{Y}^{\ell}
		=  \mathbf{H}^{\ell} \mathbf{U}^{\ell}
		+ \mathbf{N}^{\ell},
	\end{align}
	where
	\begin{align}
		\mathbf{H}^{\ell} &=
		\begin{bmatrix}
			\mathfrak{Re}\{\mathbf{H}_{\text{c}}^{\ell}\} \\
			\mathfrak{Im}\{\mathbf{H}_{\text{c}}^{\ell}\}
		\end{bmatrix}, \quad
		\mathbf{Y}^{\ell} =
		\begin{bmatrix}
			\mathfrak{Re}\{\mathbf{Y}_{\text{c}}^{\ell}\} \\
			\mathfrak{Im}\{\mathbf{Y}_{\text{c}}^{\ell}\}
		\end{bmatrix}, \quad
		\mathbf{N}^{\ell} =
		\begin{bmatrix}
			\mathfrak{Re}\{\mathbf{N}_{\text{c}}^{\ell}\} \\
			\mathfrak{Im}\{\mathbf{N}_{\text{c}}^{\ell}\}
		\end{bmatrix}.
\end{align}
Under the above complex-noise convention, the entries of
$\mathbf N^\ell$ are independent zero-mean Gaussian variables with
variance $\sigma_n^2$, consistently with the Gaussian-MAC model in
Section~IV.

\noindent\textbf{Lattice Decoding:}
	Due to the closure property of lattices, under which any integer linear combination of lattice points remains a lattice point, the fusion center can target different integer-coefficient functions of the form $\mathbf{a}^\top \mathbf{U}^{\ell}$ characterized by the integer coefficient vector $\mathbf{a} = [a_1, \ldots, a_K]^\top$. To this end, the fusion center first equalizes the received signals using a vector $\mathbf{b} \in \mathbb{R}^{2M \times 1}$ through the operation $\mathbf{b}^\top \mathbf{Y}^{\ell}$, which can be expressed as
	\begin{align}
		\label{equal}
		\mathbf{b}^\top \mathbf{Y}^{\ell}
		= \mathbf{a}^\top \mathbf{U}^{\ell}
		+ \left(\mathbf{b}^\top \mathbf{H}^{\ell} - \mathbf{a}^\top\right)\mathbf{U}^{\ell}
		+ \mathbf{b}^\top \mathbf{N}^{\ell}.
	\end{align}
	After denormalization, the decoder output associated with the integer-coefficient function $\sum_{k=1}^{K} a_k \mathbf{x}_k^{\ell}$ is
	\begin{align}
		\mathbf{s}^{\ell}
		= \frac{{\rho}^{\,\ell-1}}{\alpha}
		Q_{\alpha \Lambda_1}\!\left(
		\mathbf{b}^\top \mathbf{Y}^{\ell}
		\right).
	\end{align}
{
The effective disturbance
in \eqref{equal} is
	\begin{align}
		\mathbf z_{\mathrm e}^{\ell}
		=
		\underbrace{
			\left(
			\mathbf b^{\top}\mathbf H^{\ell}-\mathbf a^{\top}
			\right)\mathbf U^{\ell}
		}_{\mathbf r_{\mathrm e}^{\ell}
			(\mathbf a,\mathbf b,\mathbf U^{\ell})}
		+
		\mathbf b^{\top}\mathbf N^{\ell}.
	\end{align}
	Hence, it contains a finite-constellation residual-interference term
	and is not generally Gaussian. Conditioned on
	$\mathbf U^{\ell}$ and $\mathbf H^{\ell}$, however, the residual
	vector $\mathbf r_{\mathrm e}^{\ell}$ is deterministic and
	\begin{align}
		\mathbf b^{\top}\mathbf N^{\ell}
		\sim
		\mathcal N\!\left(
		\mathbf 0,
		\nu_{\mathbf b}\mathbf I
		\right),
		\qquad
		\nu_{\mathbf b}
		\triangleq
		\sigma_n^2\|\mathbf b\|^2.
	\end{align}
	The exact conditional finite-dimensional decoding error probability
	is therefore
	\begin{align}
		&P_e^{\ell}
		\!\left(
		\mathbf a,\mathbf b
		\,\middle|\,
		\mathbf U^{\ell},\mathbf H^{\ell}
		\right)= 1-
		\nonumber\\
		&\quad
		\frac{1}{(2\pi\nu_{\mathbf b})^{n/2}}
		\int_{\alpha\mathcal V_1}
		\exp\!\left(
		-\frac{
			\left\|
			\mathbf z-
			\mathbf r_{\mathrm e}^{\ell}
			(\mathbf a,\mathbf b,\mathbf U^{\ell})
			\right\|^2}
		{2\nu_{\mathbf b}}
		\right)
		\,\mathrm d\mathbf z.
		\label{eq:exact_pe_fading_conditional}
	\end{align}
Averaging \eqref{eq:exact_pe_fading_conditional} over the
finite-constellation distribution of $\mathbf U^{\ell}$ gives the
exact decoding-error probability conditioned on $\mathbf H^{\ell}$;
equivalently, the effective disturbance follows a Gaussian mixture
with constellation-dependent means and common covariance
$\nu_{\mathbf b}\mathbf I$.
	
	For tractable equalizer and coefficient design, define the
	effective-noise metric
	\begin{align}
		\label{effective}
		\sigma_{\mathrm e}^{2}(\mathbf a,\mathbf b)
		=
		P\left\|
		\mathbf b^{\top}\mathbf H^{\ell}-\mathbf a^{\top}
		\right\|^{2}
		+
		\sigma_n^{2}\|\mathbf b\|^{2}.
	\end{align}
	We relate this metric to decoding reliability through the isotropic
	Gaussian surrogate
	\begin{align}
		\widetilde{\mathbf z}_{\mathrm e}^{\ell}
		\sim
		\mathcal N\!\left(
		\mathbf 0,
		\sigma_{\mathrm e}^{2}(\mathbf a,\mathbf b)\mathbf I
		\right).
	\end{align}
	Under this surrogate, \eqref{eq:exact_pe} gives the corresponding
	model-based decoding-error probability. Minimizing
	$\sigma_{\mathrm e}^{2}(\mathbf a,\mathbf b)$ therefore provides a
	tractable reliability-oriented design criterion. The following theorem
	gives the equalization vector that minimizes this metric.
}
	\begin{theorem}
		The optimal equalization vector for a given coefficient vector $\mathbf{a}$ is
		\begin{align}
			\label{bopt}
			\mathbf{b}_\text{opt}^\top = \mathbf{a}^\top {\mathbf{H}^\ell}^\top \left(\frac{1}{\text{SNR}}\mathbf{I}+{\mathbf{H}^\ell}{\mathbf{H}^\ell}^\top\right)^{-1},
		\end{align}
		where $\text{SNR} = P/\sigma_n^2$.
	\end{theorem}
	\begin{IEEEproof}
		Expanding $\sigma_\text{e}^2(\mathbf{a}, \mathbf{b})$ as
		\begin{align}
			&P\|\mathbf{b}^\top{\mathbf{H}^\ell}-\mathbf{a}^\top\|^2+\sigma_n^2\|\mathbf{b}\|^2 = P\left(\mathbf{b}^\top{\mathbf{H}^\ell}-\mathbf{a}^\top\right)\left({\mathbf{H}^\ell}^\top\mathbf{b}-\mathbf{a}\right)\nonumber\\&+\sigma_n^2\|\mathbf{b}\|^2 = P\mathbf{b}^\top{\mathbf{H}^\ell}{\mathbf{H}^\ell}^\top\mathbf{b} -2P\mathbf{b}^\top{\mathbf{H}^\ell}\mathbf{a}+P\mathbf{a}^\top\mathbf{a}+\sigma_n^2\mathbf{b}^\top\mathbf{b},
		\end{align}
		and taking the derivative with respect to $\mathbf{b}$, we obtain
		\begin{align}
			2P{\mathbf{H}^\ell}{\mathbf{H}^\ell}^\top \mathbf{b} - 2P {\mathbf{H}^\ell}\mathbf{a}+{2}\sigma_n^2\mathbf{b},
		\end{align}
		which amounts to zero at \eqref{bopt}.
	\end{IEEEproof}

			Substituting $\mathbf{b}_{\mathrm{opt}}$ into \eqref{effective} gives
			\begin{align}
				\label{dmse_mil}
				\sigma_{\mathrm e}^{2}(\mathbf a)
				&=
				P\,\mathbf a^{\top}
				\left[
				\mathbf I-
				{\mathbf H^{\ell}}^{\top}
				\left(
				\frac{1}{\mathrm{SNR}}\mathbf I+
				\mathbf H^{\ell}{\mathbf H^{\ell}}^{\top}
				\right)^{-1}
				\mathbf H^{\ell}
				\right]
				\mathbf a.
			\end{align}
			Applying the matrix inversion identity and defining
			\begin{align}
				\mathbf Q^{\ell}
				\triangleq
				\left(
				\mathbf I+
				\mathrm{SNR}\,
				{\mathbf H^{\ell}}^{\top}\mathbf H^{\ell}
				\right)^{-1},
			\end{align}
			we obtain
			\begin{align}
				\label{dmsee}
				\sigma_{\mathrm e}^{2}(\mathbf a)
				=
				P\,\mathbf a^{\top}\mathbf Q^{\ell}\mathbf a.
			\end{align}
		In the sequel, the dependence of $\sigma_{\mathrm e}^{2}(\cdot)$ and the optimized coefficient vectors on the layer index $\ell$ is suppressed when clear from context. All layer-wise reliability conditions are understood to hold for every $\ell\in\{1,\ldots,L\}$.

Directly decoding the target function
$\sum_{k=1}^{K}\mathbf{x}_{k}^{\ell}$ corresponds to
$\mathbf a=\mathbf 1$ and has the surrogate effective-noise metric
\begin{align}
	\sigma_{\mathrm e}^{2}(\mathbf 1)
	=
	P\,\mathbf 1^{\top}\mathbf Q^{\ell}\mathbf 1.
\end{align}

{
Specializing the Gaussian surrogate to $\mathbf a=\mathbf 1$, the layer-wise model-based decoding-error probability is
\begin{align}
	P_{e,\mathrm d}^{\ell}
	=
	\mathcal{E}_{\alpha\Lambda_{1}}
	\!\left(\sigma_{\mathrm e}^{2}(\mathbf 1)\right).
\end{align}
For a prescribed end-to-end error budget $\epsilon_{\mathrm d}$,
applying the union bound over the $L$ layer-wise decisions gives the
sufficient condition
\begin{align}
	\sigma_{\mathrm e}^{2}(\mathbf 1)
	\leq
	\sigma_{\mathrm{rel}}^{2}
	\!\left(\frac{\epsilon_{\mathrm d}}{L}\right).
	\label{eq:direct_reliable_exact}
\end{align}
}

We refer to this approach, in which the target function is decoded directly as a single integer-coefficient function, as \emph{direct computation}.

		\vspace{-5pt}
	\subsection{Collective Out-of-Air Computation}
	Instead of directly decoding the target function $\sum_{k=1}^{K} \mathbf{x}_k^{\ell}$---corresponding to the coefficient vector $\mathbf{a}=\mathbf{1}$---the fusion center may decode a \emph{set} of integer-coefficient functions and subsequently reconstruct the desired function by combining these function representations through a structured linear system. We refer to this concept as \textit{collective computation}, where multiple function representations, decoded with higher reliability, are used to assist in decoding the target function.
	
	Let $\{\mathbf{a}^{(i)}\}_{i=1}^{I}$ denote $I$ integer coefficient vectors, collected in the matrix
\begin{align}
	\mathbf{A} \triangleq \big[\mathbf{a}^{(1)},\ldots,\mathbf{a}^{(I)}\big] \in \mathbb{Z}^{K\times I},
\end{align}
	where the coefficient vector $\mathbf{a}^{(i)} = [a_1^{(i)}, \ldots,a_K^{(i)}]^\top$ corresponds to the integer-coefficient function $\mathbf{s}^{\ell,(i)} = \sum_{k=1}^{K} a_k^{(i)} \mathbf{x}_k^\ell$. The desired function $\mathbf{s}^\ell = \sum_{k=1}^{K} \mathbf{x}_k^{\ell}$ can be recovered from the decoded function representations if there exists a vector
\begin{align}
	\mathbf{c} = [c_1,\ldots,c_I]^\top \in \mathbb{R}^{I},
\end{align}
	such that
	\begin{align}
		\mathbf{A}\mathbf{c} = \mathbf{1}.
		\label{eq:Ac1_constraint}
	\end{align}
	
	The equation \eqref{eq:Ac1_constraint} is solvable if and only if
	$\mathbf 1$ lies in the column space of $\mathbf A$.
	
	{
		Under the Gaussian surrogate introduced above, the layer-wise
		model-based decoding-error probability of the $i$-th function is
		\begin{align}
			P_{e,\mathrm c}^{\ell,(i)}
			=
			\mathcal E_{\alpha\Lambda_1}
			\!\left(
			\sigma_{\mathrm e}^{2}\big(\mathbf a^{(i)}\big)
			\right).
		\end{align}
	Let $I_\ell$ denote the number of functions selected at layer $\ell$
	and define the total number of collective lattice decisions as $N_{\mathrm c}
		\triangleq
		\sum_{\ell=1}^{L}I_\ell$.
	Thus, for the generic layer considered above, the local notation
	$I$ equals $I_\ell$.
	
	At a given layer, the corresponding effective-noise vectors are
	formed from the same transmitted matrix $\mathbf U^\ell$ and noise
	matrix $\mathbf N^\ell$ and are therefore generally dependent. For a
	prescribed end-to-end error budget $\epsilon_{\mathrm c}$, applying
	the union bound over the $N_{\mathrm c}$ required decisions gives the
	sufficient conditions
	\begin{align}
		\label{fading_reliable}
		\sigma_{\mathrm e}^{2}\!\big(\mathbf a^{(i)}\big)
		\leq
		\sigma_{\mathrm{rel}}^{2}
		\!\left(\frac{\epsilon_{\mathrm c}}{N_{\mathrm c}}\right),
		\qquad
		i=1,\ldots,I_\ell,
	\end{align}
	at every layer. Hence, collective reliability is governed by the
	largest effective-noise metric among the selected functions, whose
	admissible value is determined explicitly by the desired end-to-end
	error probability and the lattice geometry. This motivates the
	design metric $\max_{i\in\{1,\ldots,I\}}
	\sigma_{\mathrm e}^{2}\big(\mathbf a^{(i)}\big)$.}
	The design of an indirect decoding strategy that seeks to improve
	upon direct decoding with $\mathbf a=\mathbf1$ can therefore be
	formulated as the following min--max optimization problem:
	\begin{align}
		\label{eq:indirect_minmax}
		\underset{I \in \{1,\ldots,K\},\ \mathbf A\in\mathbb Z^{K\times I}\setminus\{\mathbf0\},\,\mathbf{c}\in\mathbb{R}^{I}}{\text{min}} 
		\quad & \max_{i\in\{1,\ldots,I\}} \; \sigma_{\mathrm{e}}^{2}\big(\mathbf{a}^{(i)}\big) \\
		\text{subject to} 
		\quad & \mathbf{A}\mathbf{c} = \mathbf{1}. \nonumber
	\end{align}
	
{
	The collective coefficient-selection problem is nonconvex,
	combinatorial, and NP-hard in general. For an integer $a_{\max}\geq1$, restricting every coefficient to
	\[
	\mathcal A\triangleq\{-a_{\max},\ldots,a_{\max}\},
	\qquad
	B\triangleq|\mathcal A|=2a_{\max}+1,
	\]
	yields an exact finite problem. A naive exact search over all
	$I\in\{1,\ldots,K\}$ has worst-case complexity
	$\mathcal O(K^{3}B^{K^{2}})$.
	
	To reduce this complexity, we propose an exact bounded-global search.
	Let $J_{\mathrm d}\triangleq\sigma_{\mathrm e}^{2}(\mathbf 1)$.
	Since direct decoding is feasible, only vectors satisfying
	$\sigma_{\mathrm e}^{2}(\mathbf a)\leq J_{\mathrm d}$ need to be
	considered. Ordering these vectors by their effective-noise metric values,
	the first threshold $\gamma$ satisfying
	\[
	\mathbf 1\in
	\operatorname{span}\!\left\{
	\mathbf a\in\mathcal A^{K}\setminus\{\mathbf 0\}:
	\sigma_{\mathrm e}^{2}(\mathbf a)\leq\gamma
	\right\}
	\]
	gives the exact bounded optimum over all $I\in\{1,\ldots,K\}$. Examining at
	most $B^{K}-1$ individual vectors gives worst-case complexity $\mathcal O\!\left(K^{2}B^{K}+B^{K}\log B^{K}\right)$
	per layer. Since this dependence remains exponential in $K$, we
	next introduce a polynomial-complexity two-group approximation.
}

For the unit vector $\mathbf e_k$,
	\eqref{dmsee} gives
	\begin{align}
		\sigma_{\mathrm e}^{2}(\mathbf e_k)
		=
		PQ_{kk}^{\ell}.
	\end{align}
	Thus, $Q_{kk}^{\ell}$ indicates the relative reliability of device
	$k$ at layer $\ell$, where a smaller value corresponds to a more
	reliable contribution.
	
	\medskip
	
\noindent\textbf{Step 1 (Channel-aware grouping):} At layer $\ell$, the $K$ devices are partitioned into two groups as
\begin{align}
	k\in\mathcal K_1^{\ell} \;\;\text{if}\;\;
	Q_{kk}^{\ell}\geq\tau^{\ell}, \qquad
	\mathcal K_2^{\ell}
	=
	\{1,\ldots,K\}\setminus\mathcal K_1^{\ell},
\end{align}
where $\tau^{\ell}$ is chosen between adjacent distinct ordered
values of $\{Q_{kk}^{\ell}\}_{k=1}^{K}$, yielding at most
$K-1$ nonempty candidate partitions.
This grouping separates devices into two sets with relatively
larger and smaller effective-noise contributions at layer $\ell$.
Since the proposed structure assigns a common coefficient within
each group, grouping devices with similar reliability helps limit
the mismatch introduced by this constraint.\medskip

 \noindent\textbf{Step 2 (Two-group coefficient structure):} All devices within the same group share a common integer coefficient. Thus, for the $i$-th decoded function ($i=1,2$), we define \begin{align} a^{(i)}_k = \begin{cases} a^{(i)}_1, & k\in\mathcal K_1^{\ell},\\ a^{(i)}_2, & k\in\mathcal K_2^{\ell}. \end{cases} \end{align}
Accordingly, each coefficient vector is parameterized by two integers, and the two decoded functions are specified by
\[
a^{(1)}_1,\, a^{(1)}_2,\quad
a^{(2)}_1,\, a^{(2)}_2.
\]

\medskip

\noindent\textbf{Step 3 (Feasibility and reconstruction):}  
The desired function can be reconstructed if there exist $c_1,c_2 \in \mathbb{R}$ such that
\begin{align}
	c_1\mathbf{a}^{(1)} + c_2\mathbf{a}^{(2)} = \mathbf{1},
\end{align}
which reduces to
\begin{align}
	c_1 a^{(1)}_1 + c_2 a^{(2)}_1 &= 1, \\
	c_1 a^{(1)}_2 + c_2 a^{(2)}_2 &= 1.
\end{align}
A unique solution exists if
\begin{align}
\det\!\begin{bmatrix}
	a^{(1)}_1 & a^{(2)}_1 \\
	a^{(1)}_2 & a^{(2)}_2
\end{bmatrix}
\neq 0.
\end{align}

\medskip

\noindent\textbf{Step 4 (Low-complexity search):}  
The integer coefficients are selected from $\mathcal{A}$ by solving
\begin{align}
	\label{eq:two_group_minmax} 
\underset{
	a^{(1)}_1,a^{(1)}_2,\,
	a^{(2)}_1,a^{(2)}_2\in{\cal A}
}{\text{min}}
	\max\{\sigma_{\mathrm{e}}^{2}(\mathbf{a}^{(1)}),\,\sigma_{\mathrm{e}}^{2}(\mathbf{a}^{(2)})\}
\end{align}
subject to the determinant constraint.

\medskip

{
	For the proposed structure, let
	$\mathbf S_{\tau}^{\ell}\in\{0,1\}^{K\times 2}$ contain, as its
	columns, the indicator vectors of $\mathcal K_1^{\ell}$ and
	$\mathcal K_2^{\ell}$, and define
	\begin{align}
		\mathbf R_{\tau}^{\ell}
		\triangleq
		(\mathbf S_{\tau}^{\ell})^{\top}
		\mathbf Q^{\ell}
		\mathbf S_{\tau}^{\ell}.
	\end{align}
	Writing $\mathbf z^{(i)}=[a_1^{(i)},a_2^{(i)}]^{\top}$ gives
	$\mathbf a^{(i)}=\mathbf S_{\tau}^{\ell}\mathbf z^{(i)}$ and
	\[
	\sigma_{\mathrm e}^{2}(\mathbf a^{(i)})
	=
	P(\mathbf z^{(i)})^{\top}
	\mathbf R_{\tau}^{\ell}\mathbf z^{(i)}.
	\]
For a fixed partition, the collective two-group search examines at
most $B^4$ coefficient tuples. Forming
$\mathbf R_{\tau}^{\ell}$ requires $\mathcal O(K^2)$ operations,
after which each objective evaluation and determinant test has
constant complexity. The search is repeated over the at most
$K-1$ candidate partitions, and the partition yielding the smallest
objective value in \eqref{eq:two_group_minmax} is selected.
The resulting worst-case complexity is $\mathcal O(K^3+KB^4)$
per layer.
}

	Once a feasible coefficient matrix $\mathbf A=[\mathbf a^{(1)},\ldots,\mathbf a^{(I)}]$ is selected, potentially from \eqref{eq:indirect_minmax}, the fusion center computes $\mathbf b^{(i)}=\mathbf b_{\mathrm{opt}}(\mathbf a^{(i)})$ for each $\mathbf a^{(i)}$ according to~\eqref{bopt}, and decodes the corresponding integer-coefficient function as
	\begin{align}
		\mathbf s^{\ell,(i)}
		=\frac{\rho^{\ell-1}}{\alpha}
		Q_{\alpha\Lambda_1}\!\big((\mathbf b^{(i)})^\top\mathbf Y^\ell\big),
		\qquad i=1,\ldots,I.
	\end{align}
	The target function at layer~$\ell$ is subsequently reconstructed by linearly combining the decoded parts,
	\begin{align}
		\mathbf s^\ell=\sum_{i=1}^{I} c_i\,\mathbf s^{\ell,(i)}.
	\end{align}
	Summing over all layers gives $\mathbf s=\sum_{\ell=1}^{L}\mathbf s^\ell$, and after subtracting the dithers and applying the inverse mapping, the final estimate of the target function is obtained as
	\[
	\hat{\mathbf u}=\phi^{-1}\!\Big(\mathbf s-\sum_{k=1}^{K}\mathbf d_k\Big).
	\]
	
If all $N_{\mathrm c}$ layer-wise decisions satisfy
\eqref{fading_reliable}, the union bound gives a model-based
end-to-end decoding-error probability no larger than
$\epsilon_{\mathrm c}$. When decoding errors are negligible, the MSE is dominated by the source-quantization distortion characterized in Section~IV. For identity normalization, 
\begin{align}
	\mathrm{MSE}
	=
	\mathbb E\!\left[
	\|\hat{\mathbf u}-\mathbf u\|^2
	\right]
	\approx
	nK\sigma^2(\Lambda_1).
\end{align}
	
The collective AirCPU procedure is illustrated in Fig. 4.
	
	
	{\noindent\textbf{Out-of-Air Principle:}
	AirCPU treats the imperfect, channel-imposed received signal as a \emph{rich knowledge source} from which valuable and meaningful computational structures (i.e., function representations) can be mined. 
	From this single observation, the receiver can \emph{extract} multiple reliable and structured functions, each conveying a distinct and complementary piece of information.  
	By freely harvesting and assembling these extracted functions, the fusion center reconstructs the desired target function directly \emph{out of the air}.} 

\begin{figure}[t]
	\centering
	\resizebox{1.0\columnwidth}{!}{%
			\begin{tikzpicture}[
				font=\small,
				>=Latex,
				node distance=6mm and 8mm,
				blk/.style={draw, rounded corners, align=center, minimum height=7.5mm, minimum width=23mm},
				sblk/.style={draw, rounded corners, align=center, minimum height=6.8mm, minimum width=20mm},
				big/.style={draw, rounded corners, align=center, minimum height=18mm, minimum width=26mm},
				line/.style={-Latex, thick}
				]
				
				\newcommand{\MiniConstellation}{%
					\tikz[line cap=round,line join=round,scale=0.18]{
						\def\s{0.8}
						\pgfmathsetmacro{\rt}{sqrt(3)}
						
						\def\XY##1##2{%
							\pgfmathsetmacro{\X}{1.5*(##1)*\s}%
							\pgfmathsetmacro{\Y}{\rt*((##2)+(##1)/2)*\s}%
						}
						
						\def\Hex##1##2##3##4{%
							\begin{scope}[shift={(##2,##3)}]
								\draw[dashed,##1] (0:##4)--(60:##4)--(120:##4)--(180:##4)--(240:##4)--(300:##4)--cycle;
							\end{scope}
						}
						
						\def\lwF{0.45pt}
						\def\lwM{1.4pt}
						\def\dotR{0.18}
						
						\def\CircleDotAt##1##2{%
							\fill[black] (##1,##2) circle[radius=\dotR];
						}
						
						\def\SmallHexDot##1##2{%
							\XY{##1}{##2}%
							\Hex{black,line width=\lwF}{\X}{\Y}{\s}%
							\CircleDotAt{\X}{\Y}%
						}
						
						\Hex{blue!70!black,line width=\lwM}{0}{0}{3*\s}
						
						\begin{scope}
							\clip (0:3*\s)--(60:3*\s)--(120:3*\s)--(180:3*\s)--(240:3*\s)--(300:3*\s)--cycle;
							
							\foreach \q in {-8,...,8}{
								\foreach \r in {-8,...,8}{
									\XY{\q}{\r}
									\Hex{black,line width=\lwF}{\X}{\Y}{\s}
									
									\pgfmathtruncatemacro{\skipdot}{%
										(\q==-1 && \r== 2) ||%
										(\q== 1 && \r== 1) ||%
										(\q==-1 && \r==-1) ||%
										(\q== 1 && \r==-2) ? 1 : 0}
									\ifnum\skipdot=0
									\CircleDotAt{\X}{\Y}
									\fi
								}
							}
						\end{scope}
						
						\SmallHexDot{-2}{1}
						\SmallHexDot{ 2}{-1}
						
						\Hex{blue!70!black,line width=\lwM}{0}{0}{3*\s}
						
						\draw[dashed,green!80!black,line width=1pt] (0,0) circle[radius=3*\s];
					}%
				}
				
			\node[sblk] (d1) {\textbf{Device $1$}\\
				$\mathbf x_1=Q_{\Lambda_1}
				(\phi(\mathbf u_1)+\mathbf d_1)$\\
				$\mathbf x_1^\ell=
				Q_{\Lambda_\ell}(\mathbf x_1)
				-Q_{\Lambda_{\ell+1}}(\mathbf x_1)$\\
				$\mathbf u_1^\ell=
				\frac{\alpha}{\rho^{\ell-1}}\mathbf x_1^\ell$};
				
				\node[sblk, below=3mm of d1] (dk) {\textbf{Device $k$}\\
					$\mathbf x_k=Q_{\Lambda_1}
					(\phi(\mathbf u_k)+\mathbf d_k)$\\
					$\mathbf x_k^\ell=
					Q_{\Lambda_\ell}(\mathbf x_k)
					-Q_{\Lambda_{\ell+1}}(\mathbf x_k)$\\
					$\mathbf u_k^\ell=
					\frac{\alpha}{\rho^{\ell-1}}\mathbf x_k^\ell$};
				
				\node[sblk, below=3mm of dk] (dK) {\textbf{Device $K$}\\
					$\mathbf x_K=Q_{\Lambda_1}
					(\phi(\mathbf u_K)+\mathbf d_K)$\\
					$\mathbf x_K^\ell=
					Q_{\Lambda_\ell}(\mathbf x_K)
					-Q_{\Lambda_{\ell+1}}(\mathbf x_K)$\\
					$\mathbf u_K^\ell=
					\frac{\alpha}{\rho^{\ell-1}}\mathbf x_K^\ell$};
				
				\node[draw, rounded corners, inner sep=3mm, fit=(d1)(dK)] (txbox) {};
				
				\node[big, right=22mm of dk] (chan) {\textbf{Fading MAC}\\[1mm]
					$\mathbf{H}^\ell$};
				
				\node[inner sep=0pt] (hex1) at ($(d1.east)+(10mm,0)$) {\MiniConstellation};
				\node[inner sep=0pt] (hexk) at ($(dk.east)+(10mm,0)$) {\MiniConstellation};
				\node[inner sep=0pt] (hexK) at ($(dK.east)+(10mm,0)$) {\MiniConstellation};
				
				\draw[line] (d1.east) -- (hex1.west);
				\draw[line] (dk.east) -- (hexk.west);
				\draw[line] (dK.east) -- (hexK.west);
				
				\draw[line] (hex1.east) -- node[pos=0.25, above=0.6mm] {$\mathbf{u}_1^\ell$} ++(3mm,0) |- (chan.west);
				\draw[line] (hexk.east) -- node[pos=0.25, above=0.6mm] {$\mathbf{u}_k^\ell$} ++(3mm,0) -- (chan.west);
				\draw[line] (hexK.east) -- node[pos=0.25, below=0.6mm] {$\mathbf{u}_K^\ell$} ++(3mm,0) |- (chan.west);
				
				\node[draw, circle, inner sep=1.2pt, right=7mm of chan] (sum) {$\sum$};
				\node[below=6mm of sum] (noise) {$\mathbf{N}^\ell$};
				
				\draw[line] (chan.east) -- (sum.west);
				\draw[line] (noise.north) -- (sum.south);
				
				\node[blk, right=10mm of sum] (choose) {Choose coefficients\\
					$\{\mathbf{a}^{(i)}\}_{i=1}^I$};
				
				\node[blk, below=of choose] (bopt) {Equalization\\
					$\mathbf{b}^{(i)}=\mathbf{b}_{\rm opt}(\mathbf{a}^{(i)})$};
				
				\node[blk, below=of bopt] (q) {Lattice decoding\\
					$\mathbf{s}^{\ell,(i)}=\dfrac{\rho^{\ell-1}}{\alpha}\,
					Q_{\alpha\Lambda_1}\!\big((\mathbf{b}^{(i)})^\top\mathbf{Y}^\ell\big)$};
				
				\node[blk, below=of q] (recon) {Find $\mathbf{c}$: $\mathbf{A}\mathbf{c}=\mathbf{1}$\\
					$\mathbf{s}^\ell=\sum_{i=1}^I c_i\,\mathbf{s}^{\ell,(i)}$};
				
				\node[blk, below=of recon] (agg) {Layer aggregation\\
					$\mathbf{s}=\sum_{\ell=1}^L \mathbf{s}^\ell$};
				
				\node[blk, below=of agg] (final) {De-dither \& inverse map\\
					$\widehat{\mathbf{u}}=\phi^{-1}\!\Big(\mathbf{s}-\sum_{k=1}^K\mathbf{d}_k\Big)$};
				
				\draw[line] (sum.east) -- node[above] {$\mathbf{Y}^\ell$} (choose.west);
				
				\draw[line] (choose) -- (bopt);
				\draw[line] (bopt) -- (q);
				\draw[line] (q) -- (recon);
				\draw[line] (recon) -- (agg);
				\draw[line] (agg) -- (final);
				
				\node[draw, rounded corners, inner sep=5mm,
				fit=(sum)(choose)(bopt)(q)(recon)(agg)(final),
				label={[font=\small]above:{\textbf{Fusion Center}}}] (rxbox) {};
				
			\end{tikzpicture}%
		}
		\caption{Block diagram for collective AirCPU over a fading MAC.}
		\label{fig:compact_fading_mac_blocks}
				\vspace{-6pt}
	\end{figure}
	
	\vspace{-5pt}
	\section{Successive Out-of-Air Computation}
	\label{sec:succ-aircomp}
	
	In conventional digital communications, once a codeword has been successfully decoded, subtracting its contribution from the received signal can significantly facilitate subsequent decoding by reducing the remaining interference---a principle known as successive interference cancellation (SIC). An analogous idea for CF has also been studied in \cite{nazer_suc, azimi_cf2}. 
In the same spirit, we explore \emph{successive computation} in AirCPU, in which the target function $\mathbf{1}^\top\mathbf{U}^\ell$ is decoded by exploiting a previously decoded integer-coefficient function. {The available side information is the decoded lattice point
		\begin{align}
			\widehat{\boldsymbol{\lambda}}_0^\ell
			&\triangleq
			Q_{\alpha\Lambda_1}\!\left(
			\mathbf b_0^\top\mathbf Y^\ell
			\right)
			=
			\mathbf a_0^\top\mathbf U^\ell+\mathbf e_0^\ell,
		\end{align}
		where $\mathbf a_0$ denotes the side-information coefficient vector,
		$\mathbf b_0=\mathbf b_{\mathrm{opt}}(\mathbf a_0)$ is its optimized
		equalization vector given by \eqref{bopt}, and
		$\mathbf e_0^\ell\in\alpha\Lambda_1$ is the side-information
		lattice-decoding error, which is zero under correct decoding.
		
		Using the actually decoded side information, let
		$\mathbf b_{\mathrm s}$ denote the target equalization vector and
		$\beta\in\mathbb R$ the side-information combining coefficient. The
		target-decoding statistic becomes
		\begin{align}
			&\mathbf b_{\mathrm s}^\top\mathbf Y^\ell
			+\beta\widehat{\boldsymbol{\lambda}}_0^\ell
			=
			\mathbf 1^\top\mathbf U^\ell
			+\big(
			\mathbf b_{\mathrm s}^\top\mathbf H^\ell
			+\beta\mathbf a_0^\top-\mathbf 1^\top
			\big)\mathbf U^\ell
			\nonumber\\
			&\hspace{35mm}
			+\mathbf b_{\mathrm s}^\top\mathbf N^\ell
			+\beta\mathbf e_0^\ell.
		\end{align}
	Thus, a side-information decoding error propagates to the target
	decoder through the additional term $\beta\mathbf e_0^\ell$.
	For tractable receiver design, we define the nominal surrogate
	effective-noise metric obtained by setting
	\(\mathbf e_0^\ell=\mathbf 0\):
		\begin{align}
			\sigma^2_\text{e}(\mathbf{1},\mathbf b_{\mathrm s},\beta\mid\mathbf a_0)
			=
			P\big\|
			\mathbf b_{\mathrm s}^\top\mathbf H^\ell
			+\beta\mathbf a_0^\top-\mathbf 1^\top
			\big\|^2
			+\sigma_n^2\|\mathbf b_{\mathrm s}\|^2.
		\end{align}
	This nominal metric is used only to optimize the receiver coefficients. 
	In actual operation, the decoder uses the decoded side-information point
	\(\widehat{\boldsymbol{\lambda}}_0^\ell\), which includes the realized
	error \(\mathbf e_0^\ell\). Accordingly, the target function
	\(\sum_{k=1}^{K}\mathbf x_k^\ell\) is decoded as
		\begin{align}
			\mathbf s^\ell
			=
			\frac{\rho^{\ell-1}}{\alpha}
			Q_{\alpha\Lambda_1}\!\left(
			\mathbf b_{\mathrm s}^\top\mathbf Y^\ell
			+\beta\widehat{\boldsymbol{\lambda}}_0^\ell
			\right).
		\end{align}
	}
	\begin{theorem}
		\label{thm:succ-opt}
		For a fixed nonzero side-information vector $\mathbf a_0$, the
		receiver coefficients that minimize
		$\sigma^2_\text{e}(\mathbf{1},\mathbf b_\text{s},\beta\mid\mathbf a_0)$
		are
		\begin{align}
			\label{sucb}
			\mathbf{b}_\text{sopt}^\top
			= (\mathbf{1} - \beta_\text{opt} \mathbf{a}_0)^\top {\mathbf{H}^\ell}^\top
			\left(\frac{1}{\text{SNR}}\mathbf{I} + {\mathbf{H}^\ell}{\mathbf{H}^\ell}^\top\right)^{-1},
		\end{align}
		and
		\begin{align}
			\label{suca}
			\beta_\text{opt}
			= \frac{\mathbf{a}_0^\top \big( \mathbf{I} + \text{SNR}\,{\mathbf{H}^\ell}^\top{\mathbf{H}^\ell} \big)^{-1}\mathbf{1}}
			{\mathbf{a}_0^\top \big( \mathbf{I} + \text{SNR}\,{\mathbf{H}^\ell}^\top{\mathbf{H}^\ell} \big)^{-1}\mathbf{a}_0 }.
		\end{align}
	\end{theorem}
	
	\begin{IEEEproof}
		Expanding the metric gives
		\begin{align}
			&P \big(\mathbf{b}_\text{s}^\top \mathbf{H}^\ell + \beta \mathbf{a}_0^\top - \mathbf{1}^\top\big)
			\big({\mathbf{H}^\ell}^\top \mathbf{b}_\text{s} + \beta \mathbf{a}_0 - \mathbf{1}\big)
			+ \sigma_n^2 \|\mathbf{b}_\text{s}\|^2 \nonumber\\
			&= P\, \mathbf{b}_\text{s}^\top \mathbf{H}^\ell {\mathbf{H}^\ell}^\top \mathbf{b}_\text{s}
			+ 2P\, \mathbf{b}_\text{s}^\top \mathbf{H}^\ell (\beta \mathbf{a}_0 - \mathbf{1})
			+ P\beta^2 \|\mathbf{a}_0\|^2 
			\nonumber\\&- 2P\beta\, \mathbf{a}_0^\top \mathbf{1}
			+ \sigma_n^2 \mathbf{b}_\text{s}^\top \mathbf{b}_\text{s}
			+ P \mathbf{1}^\top \mathbf{1}.
		\end{align}
		Taking the derivative with respect to $\mathbf{b}_\text{s}$ and setting it to zero yields
		\begin{align}
			2P\, \mathbf{H}^\ell {\mathbf{H}^\ell}^\top \mathbf{b}_\text{s}
			+ 2P\, \mathbf{H}^\ell (\beta \mathbf{a}_0 - \mathbf{1})
			+ 2\sigma_n^2 \mathbf{b}_\text{s}
			= 0,
		\end{align}
	which yields the expression in \eqref{sucb}.
		
	Substituting $\mathbf b_{\mathrm{sopt}}$ back into the metric gives
		\begin{align}
			\sigma_{\mathrm e}^{2}
			(\mathbf 1,\beta\mid\mathbf a_0)
			&=
			P(\beta\mathbf a_0-\mathbf 1)^\top
			\mathbf Q^{\ell}
			(\beta\mathbf a_0-\mathbf 1)
			\nonumber\\
			&=
			P\beta^2\mathbf a_0^\top\mathbf Q^{\ell}\mathbf a_0
			-
			2P\beta\mathbf a_0^\top\mathbf Q^{\ell}\mathbf 1
			+
			P\mathbf 1^\top\mathbf Q^{\ell}\mathbf 1.
		\end{align}
		Differentiating with respect to $\beta$ and setting the result to zero
		gives
		\begin{align}
			2P\beta\mathbf a_0^\top\mathbf Q^{\ell}\mathbf a_0
			-
			2P\mathbf a_0^\top\mathbf Q^{\ell}\mathbf 1
			=
			0,
		\end{align}
		which leads to \eqref{suca}.
	
	\end{IEEEproof}

After substituting $\beta_{\mathrm{opt}}$, the minimum surrogate
effective-noise metric for a fixed side-information vector
$\mathbf a_0$ is
	\begin{align}
		\sigma_{\mathrm e}^{2}
		(\mathbf 1\mid\mathbf a_0)
		=
		P\left[
		\mathbf 1^\top\mathbf Q^{\ell}\mathbf 1
		-
		\frac{
			\big(
			\mathbf a_0^\top\mathbf Q^{\ell}\mathbf 1
			\big)^2}
		{\mathbf a_0^\top\mathbf Q^{\ell}\mathbf a_0}
		\right].
		\label{eq:successive_variance_compact}
	\end{align}
	The second term in \eqref{eq:successive_variance_compact} is the
	effective-noise reduction obtained by projecting the target
	coefficient vector onto the side-information direction under the
	metric induced by $\mathbf Q^{\ell}$.
	
{ Under the isotropic Gaussian surrogate introduced in Section~V, define the model-based side-information and nominal target-decoding error probabilities, respectively, as 
\begin{align} P_{e,\mathrm{side}}^{\ell} &= \mathcal E_{\alpha\Lambda_1} \!\left( \sigma_{\mathrm e}^{2}(\mathbf a_0) \right), \\ P_{e,\mathrm{target},\mathrm{nom}}^{\ell} &= \mathcal E_{\alpha\Lambda_1} \!\left( \sigma_{\mathrm e}^{2} (\mathbf 1\mid\mathbf a_0) \right), 
\end{align}
where the latter corresponds to the nominal target-decoding statistic obtained by setting $\mathbf e_0^\ell=\mathbf0$. The successive target-error event is contained in the union of these two error events. Therefore, the corresponding successive decoding-error probability satisfies 
\begin{align} P_{e,\mathrm{succ}}^\ell \leq P_{e,\mathrm{side}}^\ell + P_{e,\mathrm{target},\mathrm{nom}}^\ell. 
\end{align} 
This union bound accounts for side-information error propagation without requiring independence between the two decoding decisions.  Therefore, under the adopted Gaussian surrogate, model-based
sufficient reliability conditions for an end-to-end error budget
\(\epsilon_s\) over all \(L\) layers are
\begin{align} \sigma_{\mathrm e}^{2}(\mathbf a_0) &\leq \sigma_{\mathrm{rel}}^{2} \!\left(\frac{\epsilon_{\mathrm s}}{2L}\right), \\ \sigma_{\mathrm e}^{2} (\mathbf 1\mid\mathbf a_0) &\leq \sigma_{\mathrm{rel}}^{2} \!\left(\frac{\epsilon_{\mathrm s}}{2L}\right). \label{succ_reliable} 
\end{align} 
Outside this joint reliable regime, the propagated term $\beta\mathbf e_0^\ell$ can offset the nominal effective-noise reduction, and a successive-decoding gain is therefore not guaranteed.} Accordingly, the side-information vector $\mathbf a_0$ is selected to minimize the larger of the two effective-noise metrics, leading to the following constrained integer optimization problem:
	\begin{align}
		\label{optimalsuccessive}
	\underset{\mathbf a_0\in\mathbb Z^K\setminus\{\mathbf0\}}{\mathrm{min}}
	\quad 
	&\max\left\{\sigma_{\mathrm{e}}^{2}\big(\mathbf{1} \mid \mathbf{a}_{0}\big),\, \sigma_{\mathrm{e}}^{2}\big(\mathbf{a}_{0}\big)\right\}
	\\
	\text{subject to}
	&\quad 
	\sigma_{\mathrm{e}}^{2}\big(\mathbf{a}_{0}\big) \leq \sigma_{\mathrm{e}}^{2}\big(\mathbf{1}\big), \nonumber\\
	&\quad 
	\mathbf{a}_0 \notin \{c\mathbf{1}: c\in\mathbb{Z}\backslash\{0\}\}, \nonumber
	\end{align}
{
where the first constraint ensures that the side-information function
has an effective-noise metric no larger than that of direct target
decoding, while the second constraint excludes coefficient vectors
collinear with the target vector. Under the adopted surrogate, the
min--max objective and the first constraint jointly limit the probability
of error propagation by discouraging side-information vectors with poor
decoding reliability.
}

	The successive coefficient-selection problem is combinatorial and
	NP-hard in general. To obtain a low-complexity approximation, we
	restrict the side-information vector to the two-group structure
	introduced for collective computation. Specifically, we set
	\begin{align} a_{0,k} = \begin{cases} p, & k\in\mathcal K_1^{\ell},\\ q, & k\in\mathcal K_2^{\ell}. \end{cases} \end{align}
	where $p,q \in \mathcal A$ are integers selected from a bounded set. This restriction is tailored to the successive setting: since the target coefficient vector is $\mathbf 1$, a useful side-information vector should remain close to a uniform all-ones structure, while allowing different coefficients for the two reliability groups so as to trade off the two terms $\sigma_{\mathrm e}^2(\mathbf a_0)$ and $\sigma_{\mathrm e}^2(\mathbf 1\mid \mathbf a_0)$.
	
	The parameters $(p,q)$ are chosen by solving
	\begin{align}
		\label{twogroupsuccessive}
		\underset{p,q\in\mathcal A}{\mathrm{min}}
		\quad
		&\max\left\{\sigma_{\mathrm e}^{2}\big(\mathbf 1\mid \mathbf a_0\big),\,
		\sigma_{\mathrm e}^{2}\big(\mathbf a_0\big)\right\}
	\end{align}
	subject to
	\[
	\sigma_{\mathrm e}^{2}(\mathbf a_0)\le \sigma_{\mathrm e}^{2}(\mathbf 1), 
	\qquad p\neq q.
	\]
{
An exact bounded-global search over the original successive problem examines
at most $B^{K}-1$ nonzero side-information vectors and has
worst-case complexity $\mathcal O(K^{2}B^{K})$ per layer. For a
fixed partition, the successive two-group search examines at most
$B(B-1)\leq B^{2}$ coefficient pairs. Since
$\mathbf a_0=\mathbf S_{\tau}^{\ell}[p,q]^{\top}$ and
$\mathbf 1=\mathbf S_{\tau}^{\ell}[1,1]^{\top}$, its objective and
constraints can be evaluated from $\mathbf R_{\tau}^{\ell}$.
The search is repeated over the at most $K-1$ candidate partitions,
and the partition yielding the smallest objective value in
\eqref{twogroupsuccessive} is selected. The resulting worst-case
complexity is $\mathcal O(K^{3}+KB^{2})$ per layer.
}

	Finally, the layered outputs are combined to produce 
	\(\mathbf{s} = \sum_{\ell=1}^{L} \mathbf{s}^{\ell}\), and removing the dithering terms and performing the inverse transformation as in \eqref{final_stage} yields the final reconstruction of the target function $\hat{\mathbf{u}}$.
	
Thus, when the decoding-error probabilities associated with the
side-information and target functions are sufficiently small, the
distortion is dominated by lattice quantization. For identity
normalization, the successive-computation MSE satisfies
\begin{align}
	\mathrm{MSE}
	=
	\mathbb E\!\left[
	\|\hat{\mathbf u}-\mathbf u\|^2
	\right]
	\approx
	nK\sigma^2(\Lambda_1).
\end{align}

		Clearly, this successive approach can be extended to incorporate multiple previously decoded functions to further improve performance, which we leave for future work.
	
	{\noindent\textbf{Out-of-Air Principle:}
	Successive computation further sharpens the out-of-air viewpoint by showing that previously extracted functions are not merely intermediate outputs, but reusable computational representations, and therefore constitute \textit{resources}. Once a reliable function has been decoded, it becomes a new piece of side information that can be injected back into the decoding process to reshape the effective uncertainty seen by the target function. 
	In this way, AirCPU does not rely on a single "shot" extraction from the channel; it progressively \emph{mines} and \emph{recombines} reliable abstractions already obtained \textit{out of the air}, so that each newly recovered function makes subsequent recovery easier and more reliable.}
	
\section{{Complexity, Scalability, and Hardware Implementation}}
\label{sec:complexity_scalability}
	{
		The per-device result in Section~III-B gives
		$\mathcal{O}(KL[C_Q(n)+n])$ network-wide encoding operations per
		source block. Since all $K$ devices transmit concurrently, the data
		phase occupies $Ln$ real channel dimensions independently of $K$.
		Each receiver-side lattice decoding costs
		$C_Q(n)+\mathcal{O}(n)$, with an additional
		$\mathcal{O}(Mn)$ linear-equalization cost over a fading MAC.
		Direct, successive, and collective AirCPU require $L$, $2L$, and $N_{\mathrm c}=\sum_{\ell=1}^{L}I_\ell$ lattice-decoding operations, respectively; for the proposed two-group collective scheme, $N_{\mathrm c}=2L$. Direct computation
		requires no coefficient search. The bounded-global searches are
		exponential in $K$, whereas the two-group searches are polynomial in
		$K$ for fixed coefficient-alphabet size $B$, as derived with the
		collective and successive formulations.
	
	The parameters $\rho$ and $n$ determine the resolution--complexity tradeoff. Since $|\mathcal X_\ell|=\rho^n$ and $\mathcal R_L=L\log_2\rho$, increasing $\rho$ reduces the number of layers, and hence the processing and transmission costs proportional to $L$, for a fixed cumulative representation rate, but enlarges the per-layer constellation. Increasing $n$ increases channel use and vector processing linearly, while the lattice-processing cost is governed by $C_Q(n)$ and the constellation cardinality grows as $\rho^n$. For identity normalization, in the reliable regime, the normalized distortion satisfies $\mathrm{MSE}/(nK)\approx\sigma^2(\Lambda_1)$, while over fading MACs the effect of increasing $K$ on reliability is captured through $\mathbf Q^\ell$. Therefore, AirCPU preserves $K$-independent data-phase latency and per-device encoding complexity, while network-wide processing, receiver-side coefficient selection, and fading-dependent reliability determine its large-scale scalability.
	
	From a hardware perspective, one base-lattice quantizer can be
	time-shared across all encoder layers, while one
	$\alpha\Lambda_1$ decoder can serve all layers and coefficient vectors.
	This avoids separate lattice-processing units but introduces latency
	proportional to the number of encoder and decoder invocations. The
	per-operation hardware cost is determined primarily by $n$ and the
	adopted nearest-neighbor quantization algorithm.
}
	\vspace{-16pt}
	\section{Simulation Results}
We evaluate direct, collective, and successive AirCPU. Unless stated
otherwise, the experiments consider sum computation, $L = 6$, the collective and
successive variants employ the proposed two-group searches, and AirCPU
denotes the collective variant; in Fig.~5, it denotes direct computation. {Fig.~\ref{fig:nonlinear_functions} instead considers RMS and geometric-mean computation. The baselines are analog AirComp, SumComp \cite{saeed_sumcomp}, complement-coded digital AirComp \cite{huang_digital}, and a one-shot composite-lattice benchmark motivated by hierarchical modulation \cite{aluini,saeed}.}
	
{Except for Fig.~11,} we employ a hexagonal lattice $\Lambda_1$, applied to vectors in pairs of components, with generator matrix $\mathbf{G}_{\mathrm{hex}}
=\delta\left[\begin{smallmatrix}
	1 & \tfrac{1}{2}\\
	0 & \tfrac{\sqrt{3}}{2}
\end{smallmatrix}\right]$
\cite{fedcpu,eldar},
where $\delta>0$ controls the quantization resolution and scales the lattice. {Fig.~11 instead evaluates the proposed computation mechanisms using the
	cubic lattice $\Lambda_1=\delta\mathbb{Z}^{2}$.}

Except for the geometric-mean experiment in Fig.~\ref{fig:nonlinear_functions}, the local input data at each device are generated independently from a uniform distribution over the interval $[-1,1]$. For the geometric-mean experiment, the inputs are generated independently over a strictly positive interval. For the fading MAC, the channel coefficients are generated independently according to a zero-mean unit-variance complex Gaussian distribution. {Unless stated otherwise, synchronization and receiver
	CSI are perfect; Figs.~5 and~6 additionally examine timing offsets and
	channel-estimation errors, respectively.} All results are obtained via Monte Carlo simulations and reported in terms of MSE versus SNR.
	
	\begin{figure}[tb!]
		\vspace{-4pt}
		\centering
		\includegraphics[width =3.0in]{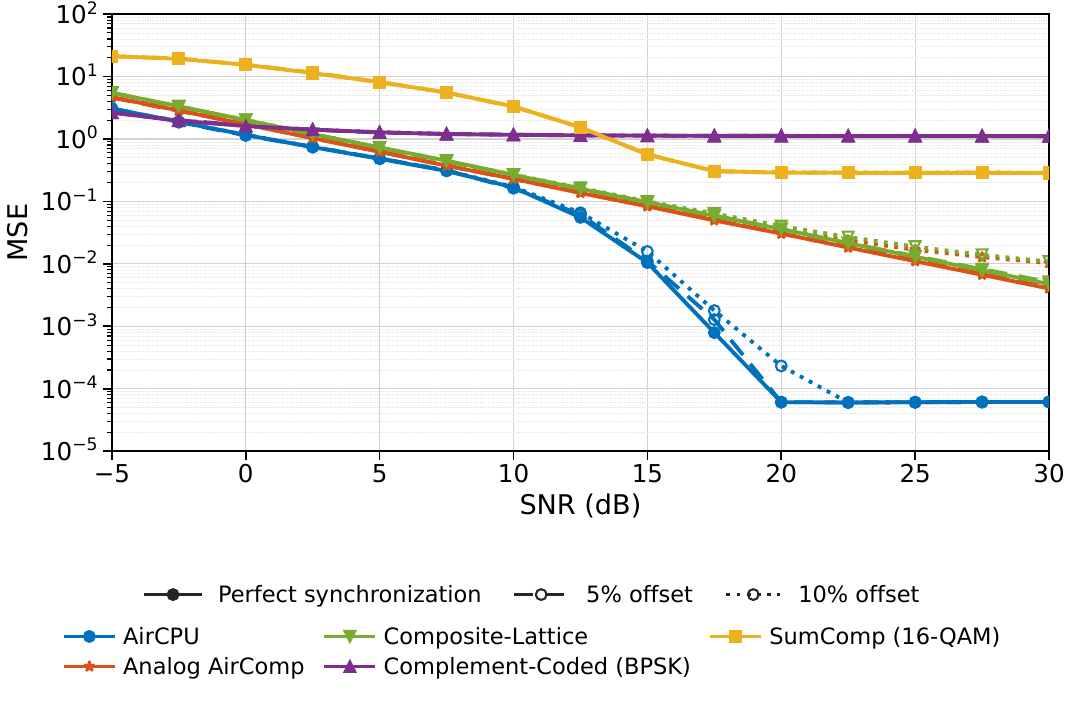} 
		\vspace{0pt}
\caption{{Performance of AirCPU and the benchmark schemes over a Gaussian MAC under perfect synchronization and residual synchronization offsets of up to $5\%$ and $10\%$ of the symbol interval, for $K=100$, $\rho=3$, and $\delta=0.001$.}}
		\vspace{-5pt}
	\end{figure}

	\begin{figure}[tb!]
		\vspace{-4pt}
		\centering
		\includegraphics[width =3.0in]{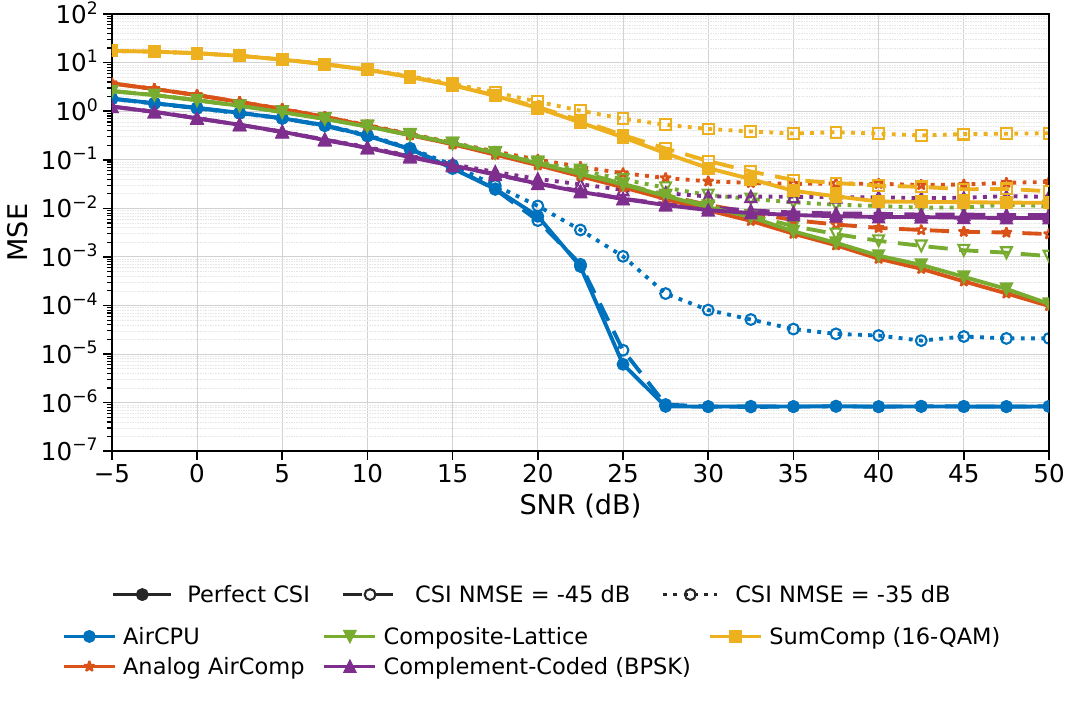} 
		\vspace{0pt}
\caption{{Performance of AirCPU and the benchmark schemes over a fading MAC under perfect receiver CSI and channel-estimation NMSEs of $-45$ and $-35$ dB, for $K=10$, $M=6$, $\rho=3$, and $\delta=0.001$.}}
		\vspace{-5pt}
	\end{figure}

	\begin{figure}[tb!]
		\vspace{-4pt}
		\centering
		\includegraphics[width =3.0in]{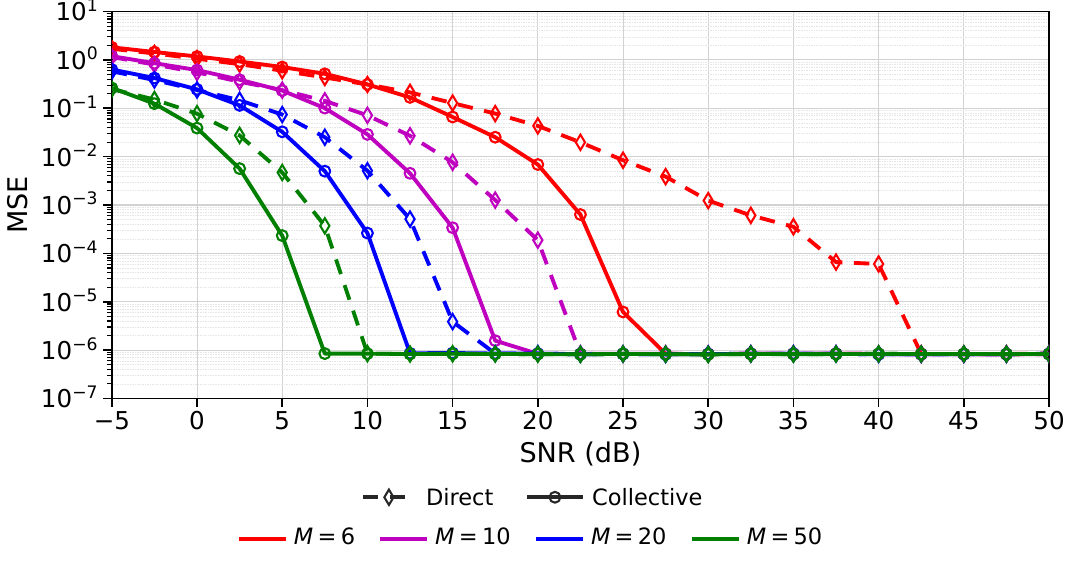} 
		\vspace{0pt}
		\caption{Performance of AirCPU, both direct and collective, over a fading MAC for different values of $M$, with $K=10$, $\rho=3$, and $\delta=0.001$.}
		\vspace{-5pt}
	\end{figure}
	
	\begin{figure}[tb!]
		\vspace{-4pt}
		\centering
		\includegraphics[width =3.0in]{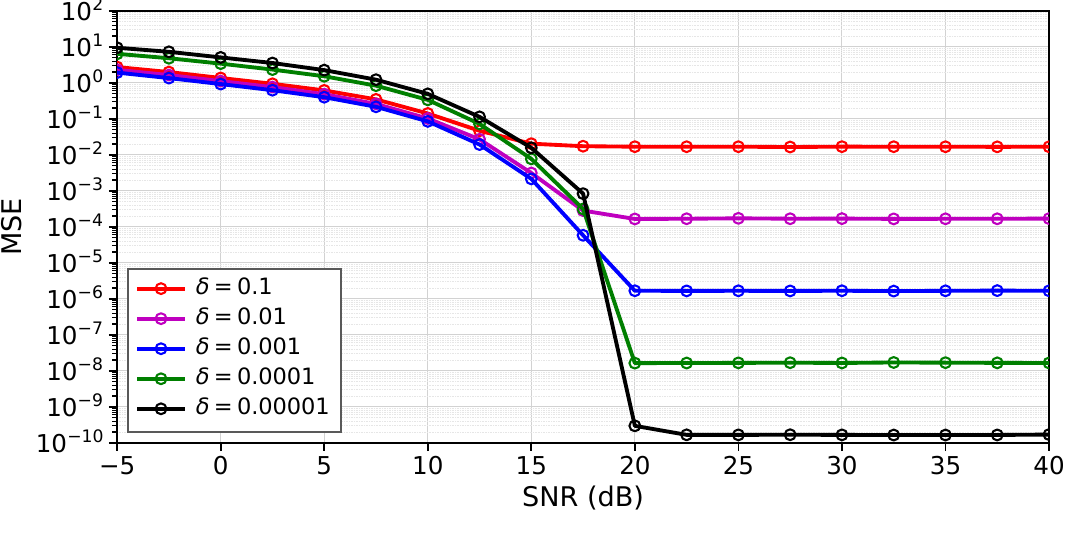} 
		\vspace{0pt}
		\caption{Performance of AirCPU over a fading MAC for different values of $\delta$, with $K=20$, $M=16$, and $\rho=3$.}
		\vspace{-5pt}
	\end{figure}

	\begin{figure}[tb!]
		\vspace{-4pt}
		\centering
		\includegraphics[width =3.0in]{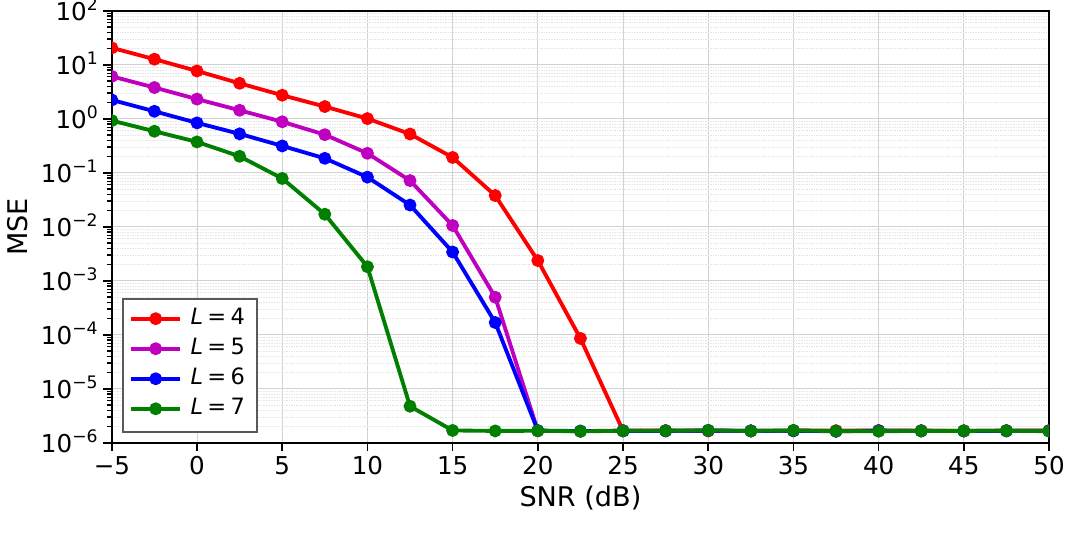} 
		\vspace{0pt}
		\caption{Performance of AirCPU over a fading MAC for different values of $L$, with $K=20$, $M=30$, and $\delta=0.001$.}
		\vspace{-5pt}
	\end{figure}

	\begin{figure}[tb!]
		\vspace{-4pt}
		\centering
		\includegraphics[width =3.0in]{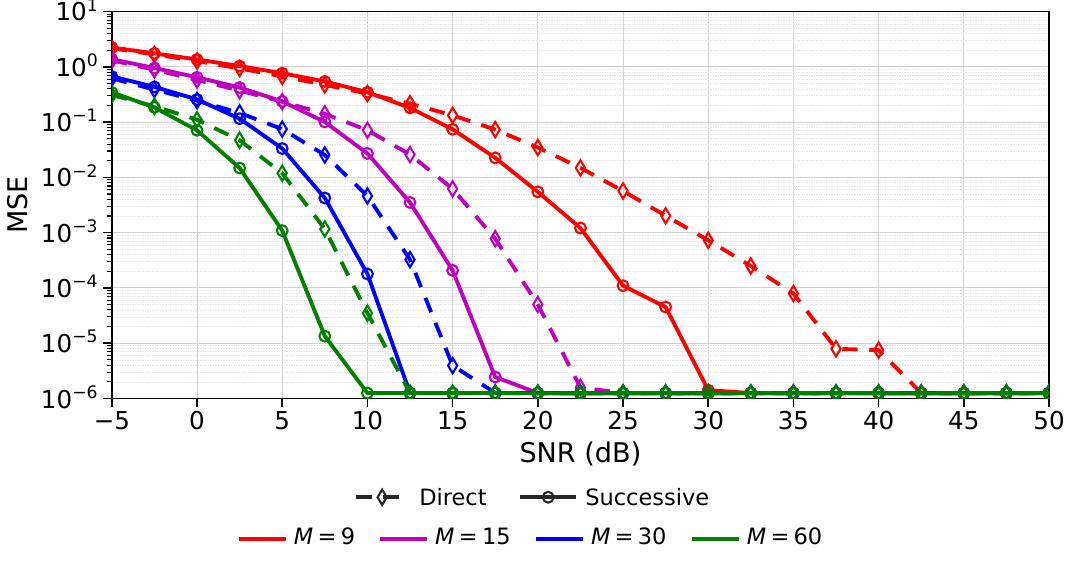} 
		\vspace{0pt}
		\caption{Performance of AirCPU, both direct and successive, over a fading MAC for different values of $M$, with $K=15$, $\rho=3$, and $\delta=0.001$.}
		\vspace{-5pt}
	\end{figure}
	
		\begin{figure}[tb!]
		\vspace{-4pt}
		\centering
		\includegraphics[width =3.0in]{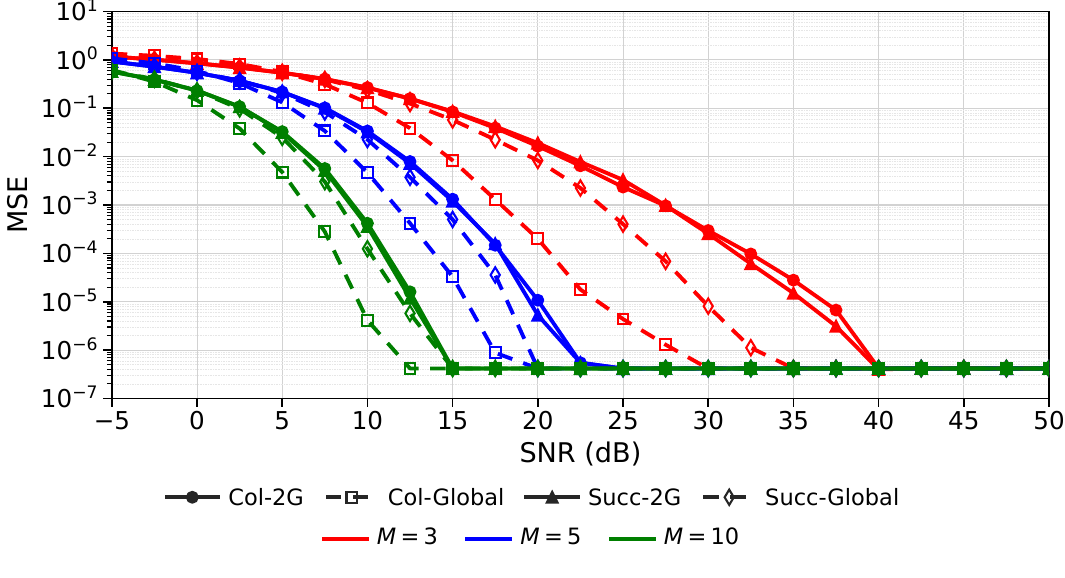} 
		\vspace{0pt}
\caption{{Performance comparison of collective and
		successive AirCPU using the proposed two-group searches and exact
		bounded-global references over a fading MAC with the cubic lattice
		$\Lambda_1=\delta\mathbb{Z}^{2}$, for $K=5$,
		$M\in\{3,5,10\}$, $\rho=3$, and $\delta=0.001$.}}
		\vspace{-5pt}
	\end{figure}
	
		\begin{figure}[tb!]
		\centering
		\includegraphics[width=3.0in]{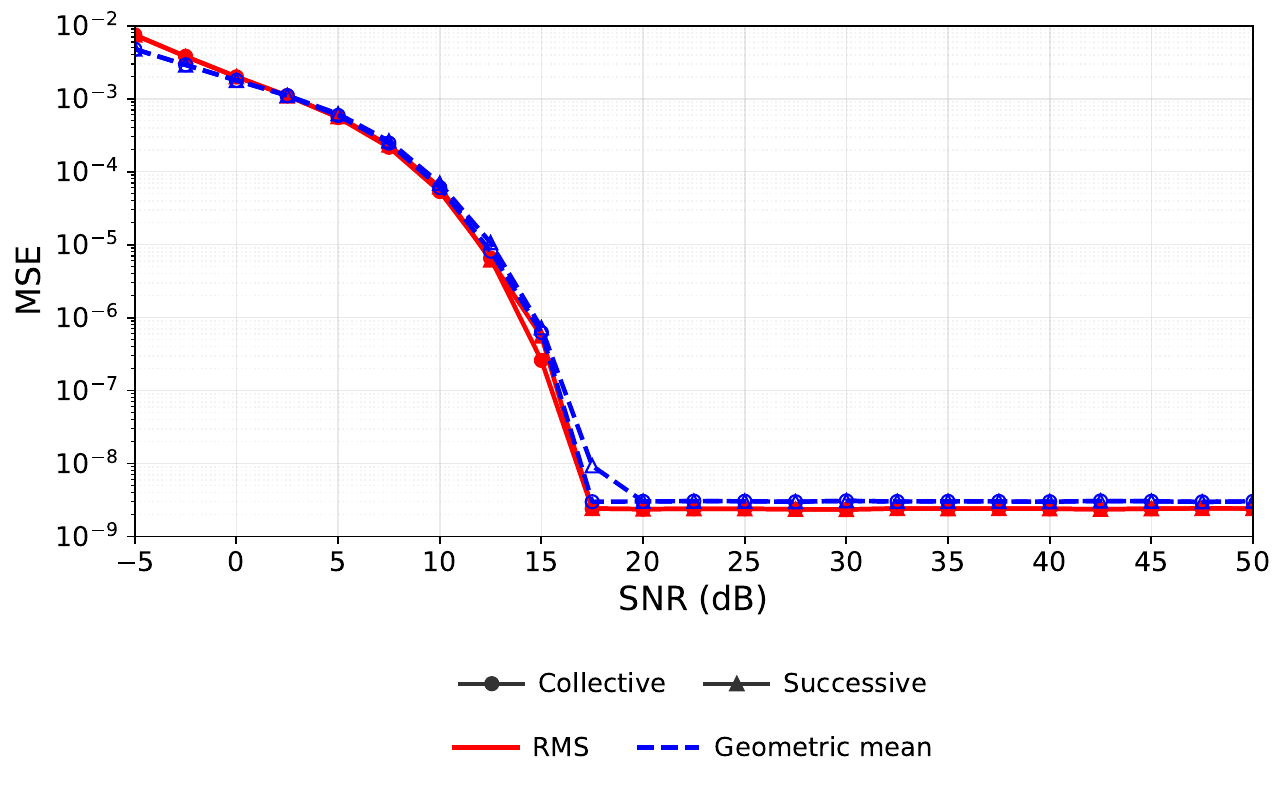}
	\caption{{Performance of collective and successive AirCPU for RMS and geometric-mean computation over a fading MAC with $K=M=15$, $\rho=3$, and $\delta=0.001$.}}
		\label{fig:nonlinear_functions}
		\vspace{-5pt}
	\end{figure}
	
{
	For a fair comparison, all schemes use the same source
	realizations, SNR definition, channel realizations,
	noise variance, MSE metric, and nominal power parameter \(P\). The
	baseline parameter choices and power normalizations follow their
	respective signal structures. AirCPU satisfies a per-use peak-energy
	constraint of \(nP\). Since the source components are uniformly
	distributed over \([-1,1]\), analog AirComp uses the scaling
	\(\sqrt{P}\), which also limits its maximum transmitted vector energy
	to \(nP\). SumComp uses one 16-QAM transmission with \(b=2\) and
	\(q=2^{2b}=16\), the configuration achieving the lowest computation
	MSE among the QAM orders evaluated in \cite{saeed_sumcomp}. Its
	constellation is scaled by \(\sqrt{6nP/(q-1)}\), yielding average
	symbol energy \(nP\) under its conventional average-energy
	normalization. Complement-coded AirComp employs a length-\(L\)
	constant-modulus BPSK codeword with per-use vector energy \(nP\).
	The composite-lattice benchmark transmits the complete \(L\)-layer
	representation in one composite symbol whose maximum transmitted
	energy is also \(nP\).
	
	Over the fading MAC, all receiver designs use the same receiver-side
	CSI and noise variance. Analog AirComp, SumComp, and the
	composite-lattice benchmark use the optimal equalizer in~\eqref{bopt} with
	\(\mathbf{a}=\mathbf{1}\), while AirCPU applies the same equalizer
	design to its selected integer coefficient vectors. Complement-coded
	AirComp uses its corresponding linear LMMSE detector under the same
	channel and noise conditions. The native transmission counts are
	retained: AirCPU and complement-coded AirComp use \(L\) transmissions,
	whereas analog AirComp, SumComp, and the composite-lattice benchmark
	use one. The comparison therefore preserves the original operating
	structure of each scheme while explicitly revealing the associated
	reliability--latency and total-energy trade-offs.
}

	{
	Fig.~5 compares these schemes over a Gaussian MAC.
	{In addition to perfect synchronization, each device
		has an independent residual timing offset uniformly distributed within
		either $\pm0.05T_s$ or $\pm0.10T_s$, where $T_s$ denotes the symbol
		interval. A raised-cosine response with roll-off factor $0.25$ is used,
		and the receiver performs no timing-offset compensation.} Under perfect synchronization, AirCPU
	exhibits a reliability transition, after which the MSE saturates at a
	floor determined primarily by the finest lattice, and further increases
	in SNR have negligible effect on the distortion.
		
		For the two-dimensional hexagonal lattice used in the simulations,
		$r_{\mathrm{in},1}=\delta/2$. Moreover, with $\rho=3$, the constellation
		radius defined earlier is $R_{2}=\sqrt{3}\delta$, and hence
		$\alpha r_{\mathrm{in},1}=\sqrt{P/6}$. Exploiting the sixfold symmetry
		of the regular hexagonal Voronoi region, the exact per-decision error
		probability in \eqref{eq:exact_pe} becomes
		\begin{align}
			P_e
			&=
			\frac{6}{\pi}
			\int_{0}^{\pi/6}
			\exp\!\left(
			-\frac{\mathrm{SNR}}
			{12\cos^{2}\theta}
			\right)
			\,\mathrm{d}\theta
			\label{eq:hex_exact_pe}=
			12T\!\left(
			\sqrt{\frac{\mathrm{SNR}}{6}},
			\frac{1}{\sqrt{3}}
			\right),
		\end{align}
		where Owen's $T$-function is defined by
		\begin{align}
			T(h,a)
			\triangleq
			\frac{1}{2\pi}
			\int_{0}^{a}
			\frac{
				\exp\!\left[-\frac{1}{2}h^{2}(1+t^{2})\right]}
			{1+t^{2}}
			\,\mathrm{d}t.
		\end{align}
		Thus, the hexagonal lattice admits an exact closed-form
		special-function characterization rather than only an inscribed-ball
		approximation. For example, solving
		\eqref{eq:hex_exact_pe} for a per-decision target
		$P_e=10^{-4}$ gives
		$\mathrm{SNR}\simeq20.1\,\mathrm{dB}$.
		The rigorous sufficient condition
		\eqref{eq:nonasymptotic_reliability} gives
		$\mathrm{SNR}\geq
		6\chi^{2}_{2,0.9999}\simeq20.4\,\mathrm{dB}$.
		The resulting $0.3\,\mathrm{dB}$ gap in the required SNR at
		$P_e=10^{-4}$ shows that the non-asymptotic upper bound is
		numerically tight at the considered operating point and is
		consistent with the reliability transition observed near
		$20\,\mathrm{dB}$ in Fig.~5.
		
	In contrast, analog AirComp remains noise-limited, while SumComp and
	complement-coded AirComp are additionally affected by their finite
	source representations. The synchronization offsets shift the AirCPU reliability transition,
	but AirCPU achieves the lowest MSE among the considered schemes under both
	offset levels in the adopted per-use-power, native-transmission setting. The one-shot composite-lattice benchmark reduces the number of
		transmissions, but its larger composite constellation is more sensitive
		to noise and synchronization mismatch, illustrating the
		reliability--latency trade-off discussed in \textit{Remark~2}.}

This observation highlights a modeling assumption in SumComp and
complement-coded AirComp, namely that they operate on discretized
symbols and implement waveform superposition over finite alphabets
rather than directly on continuous-valued inputs, as in analog
AirComp, thereby incurring additional quantization
effects. {AirCPU differs from the considered baselines
	by directly operating on continuous sources through a structured joint
	source--channel coding architecture, which leads to improved performance
	under the considered settings.}

Fig.~6 repeats the comparison over a fading MAC. {In addition to perfect receiver CSI, the
	estimate is modeled as
	$\widehat{\mathbf H}_{\mathrm c}=\mathbf H_{\mathrm c}+\mathbf E_{\mathrm h}$,
	with channel-estimation NMSEs of $-45$ and $-35$ dB. Transmission
	occurs over the true channel, while all receiver equalizers and AirCPU
	coefficient selections use the same estimated channel.}

{AirCPU retains its reliability behavior under fading and reaches
the finest quantization-determined floor. Channel-estimation errors
have a limited effect when receiver noise dominates, but produce
NMSE-dependent high-SNR floors as equalization mismatch becomes
dominant. AirCPU nevertheless retains the best performance under
both considered estimation-error levels.}

	Fig.~7 compares direct and collective computation within the AirCPU framework. By decoding multiple integer-coefficient functions that are individually more reliable and reconstructing the target function from these decoded functions, the collective variant enlarges the reliable operating region compared to direct decoding. This manifests as a leftward shift of the reliability transition and a reduced diversity requirement to reach the same distortion floor. Moreover, the results indicate that as the number of receive antennas increases, additional integer-coefficient functions become decodable at lower SNRs and can assist the target function earlier, while in the high-SNR regime the relative gain of collective decoding is more pronounced for smaller antenna numbers, since revealing additional structured function representations that can be decoded with lower effective noise significantly improves the decoding of the target function.

	Fig.~8 illustrates the impact of the finest lattice resolution. As predicted by the analysis, the achievable MSE floor is directly controlled by the lattice resolution, while the SNR threshold for reliable operation shifts accordingly. This demonstrates that, in principle, arbitrary computation resolution can be targeted by AirCPU, thereby highlighting the inherent flexibility of the proposed framework.

{
	Fig.~9 illustrates the progressive accuracy--resource trade-off.
	Each additional layer contributes $\log_2\rho$ bits per source
	dimension, requires one additional $n$-dimensional MAC transmission,
	and reduces the upper bound on the unresolved refinement term in
	\eqref{eq:progressive_distortion}. For $\rho=3$, this corresponds to
	approximately $1.585$ bits per source dimension per layer. With a fixed
	finest lattice $\Lambda_1$, the full-resolution floor remains
	approximately $nK\sigma^2(\Lambda_1)$, while increasing $L$ enlarges
	the reliable operating region, provides more intermediate resolution
	levels and a larger representable range, and incurs a channel-use cost
	that grows linearly with $L$.
}

{
	Fig.~10 compares direct and successive AirCPU, where the successive decoder uses the actually decoded side-information lattice point $\widehat{\boldsymbol{\lambda}}_0^\ell$ in every realization; hence, the reported MSE includes side-information decoding errors and their propagation to the target decoder.
} The results demonstrate consistent gains, particularly in moderate diversity regimes, where side information significantly lowers the SNR required to reach the reliable region. More specifically, when the number of receive antennas is small, successive decoding provides more pronounced improvements at higher SNRs, since previously decoded functions with lower effective noise become available and can be leveraged to further refine the target decoding. In contrast, for larger antenna numbers, successive computation becomes effective already at lower SNRs, as spatial diversity enables earlier extraction of function representations that assist the target function recovery.

{
	Fig.~11 compares the proposed two-group searches with their exact
	bounded-global references for collective and successive AirCPU using
	the cubic lattice \(\Lambda_1=\delta\mathbb{Z}^{2}\) and
	\(\mathcal A=\{-6,\ldots,6\}\). Around the reliability transition,
	the bounded-global collective reference reaches the
	resolution-determined MSE floor at a substantially lower SNR than its
	successive counterpart, although this advantage narrows as \(M\)
	increases. For both computation mechanisms, the gap between the
	two-group search and its bounded-global reference also decreases with
	\(M\). The stronger performance of collective computation reflects
	its greater structural freedom to combine \(I\leq K\) functions, at
	the cost of exponential coefficient-selection complexity and \(I\)
	lattice decodings per layer, whereas its two-group version fixes
	\(I=2\). In successive computation, both searches require two lattice
	decodings per layer, but differ in coefficient-selection complexity
	and search freedom.
}

{ Fig.~\ref{fig:nonlinear_functions} extends the proposed sum-computation framework to nonlinear nomographic functions. RMS uses local squaring and the post-processing \(s\mapsto\sqrt{s/K}\), whereas the geometric mean over positive inputs uses local logarithmic preprocessing and \(s\mapsto\exp(s/K)\). The results confirm the applicability of the collective and successive mechanisms to nonlinear nomographic functions, while the underlying AirCPU transmission and sum recovery remain unchanged. }
	
	\section{Conclusions}
{ We proposed AirCPU as an AirComp framework that reinterprets the wireless MAC as a structured function-extraction source with controllable and progressively refinable resolution. The reliability analysis established the conditions under which the lattice layers required by a target resolution can be recovered and showed that, once this reliable regime is reached, the computation distortion is governed primarily by the finest reliably recovered lattice layer, while channel noise and finite-constellation constraints affect it mainly through decoding failures. Beyond direct computation, collective and successive computation exploit decoded integer-coefficient functions as reusable structural information, thereby enlarging the reliable operating region over fading MACs without CSIT-driven power control. The associated coefficient-selection problems are addressed through an exact bounded-global search with reduced complexity for collective computation and scalable two-group approximations for collective and successive computation. The numerical results corroborated the predicted transition to a resolution-determined MSE floor, demonstrated substantial gains over the considered analog and digital AirComp baselines under the considered settings, and illustrated the extension of the framework to nonlinear nomographic functions. Overall, AirCPU suggests that wireless superposition can be exploited not only by embedding a target computation into the transmitted signals, but also by extracting and reusing structured functions at the receiver. Future directions include scalable near-optimal coefficient-selection methods, deeper successive structures, broader function classes and wireless models, and sensing and learning applications.}


\vspace{0pt}
	\begin{thebibliography}{1}
		\bibitem{nazer_general}
		B. Nazer and M. Gastpar, "Computation over multiple-access channels," \emph{IEEE Trans. Inf. Theory}, vol. 53, no. 10, pp.
		3498-3516, Oct. 2007.
		\bibitem{goldenbaum}
		M. Goldenbaum, H. Boche, and S. Stanczak, "Harnessing interference for analog function computation in wireless sensor
		networks," \emph{IEEE Trans. Signal Proc.}, vol. 61, no. 20, pp. 4893-4906, 2013.
		\bibitem{aircompsurvey}
		A. Sahin and R. Yang, "A survey on over-the-air computation," \emph{IEEE Commun. Surv. Tutor.}, vol. 25, no. 3, pp. 1877-1908,
		2023.
		\bibitem{aircompmag}
		A. Perez-Neira, M. Martinez-Gost, A. Sahin, S. Razavikia, C. Fischione, and K. Huang, "Waveforms for computing over the air: A groundbreaking approach that redefines data aggregation," \emph{IEEE Signal Process. Mag.}, vol. 42, no. 2, pp. 57-77, Mar. 2025.
		\bibitem{airflmag}
		S. M. Azimi-Abarghouyi, C. Fischione, and K. Huang, "Over-the-air federated learning: Rethinking
		edge AI through signal processing," arXiv: 2512.03719
		
		\bibitem{kaibin_aircomp}
		X. Cao, G. Zhu, J. Xu, and K. Huang, "Optimized power control for over-the-air computation in fading channels," \emph{IEEE Trans. Wireless
			Commun.}, vol. 19, no. 11, pp. 7498-7513, Nov. 2020.
		
		\bibitem{kaibin_airfl}
		G. Zhu, Y. Wang, and K. Huang, "Broadband analog aggregation for
		low-latency federated edge learning," \emph{IEEE Trans. Wireless Commun.},
		vol. 19, no. 1, pp. 491-506, Jan. 2020.
		
		\bibitem{ding_airfl}
		K. Yang, et al., "Federated learning via over-the-air computation," \emph{IEEE Trans. Wireless Commun.}, vol. 19, no. 3, pp. 2022-2035, 2020.
		\bibitem{gunduz_airfl}
		M. M. Amiri and D. Gunduz, "Federated learning over wireless fading
		channels," \emph{IEEE Trans. Wireless Commun.}, vol. 19, no. 5, pp. 3546-
		3557, 2020.
		\bibitem{deniz_mis}
		Y. Shao, D. Gunduz, and S. C. Liew, "Federated edge learning with misaligned over-the-air computation," \emph{IEEE Trans. Wireless Commun.}, vol. 21, no. 6, pp. 3951-3964, June 2022.
		
		\bibitem{kaibin_aircomp2}
		Z. Lin, Y. Gong, and K. Huang, "Distributed over-the-air computing for fast
		distributed optimization: Beamforming design
		and convergence analysis," \emph{IEEE J. Sel. Areas Commun.}, vol. 41, no. 1, pp. 274-287, Jan. 2023.
		\bibitem{azimi_wafel}
		S. M. Azimi-Abarghouyi and L. Tassiulas, "Over-the-air federated learning via weighted aggregation," \emph{IEEE Trans. Wireless Commun.}, vol. 23, no. 12, pp. 18240-18253, Dec. 2024.
		\bibitem{deniz_inference}
		C. Bian, M. Hua, and D. Gunduz, "Over-the-air inference through analog computation over multi-hop MIMO networks," \emph{IEEE Wireless Commun. Lett.}, vol. 14, no. 11, pp. 3739-3743, Nov. 2025.
		\bibitem{huang_neural}
		Y. Cang,  M. Chen, and K. Huang, "Over-the-air computation for realizing neural
		link in in-network AI architectures," \emph{IEEE Trans. Wireless Commun.}, vol. 25, pp. 2373-2388, 2026.
		\bibitem{gunduz_blind}
		M. Mohammadi Amiri, et al., "Blind federated edge learning," \emph{IEEE Trans. Wireless
			Commun.}, vol. 20, no. 8, pp. 5129-5143, Aug. 2021.
		
		\bibitem{deniz_onebit}
		G. Zhu, Y. Du, D. Gunduz, and K. Huang, "One-bit over-the-air
		aggregation for communication-efficient federated edge learning: Design
		and convergence analysis," \emph{IEEE Trans. Wireless Commun.}, vol. 20,
		no. 3, pp. 2120-2135, Mar. 2021.
		\bibitem{saeed_channelcomp}
		S. Razavikia, J. M. B. Da Silva, and C. Fischione, "ChannelComp:
		A general method for computation by communications," \emph{IEEE Trans. Commun.}, vol. 72, no. 2, pp. 692-706, 2024.
		\bibitem{saeed_sumcomp}
		S. Razavikia, J. M. B. Da Silva, and C. Fischione, "SumComp: Coding for digital over-the-air computation via the ring of integers," \emph{IEEE Trans. Commun.}, vol. 73, no. 2, pp. 752-767, Feb. 2025.
		\bibitem{saeed_blind}
		S. Razavikia, J. M. B. Da Silva, and C. Fischione, "Blind federated learning via over-the-air q-QAM," \emph{IEEE Trans. Wireless Commun.}, vol. 23, no. 12, pp. 19570-19586, Dec. 2024.
		
		\bibitem{deniz_digital}
		L. Qiao, Z. Gao, M. B. Mashhadi, and D. Gunduz, "Massive digital over-the-air computation for
		communication-efficient federated edge learning," \emph{IEEE J. Sel. Areas Commun.}, vol. 42, no. 11, pp. 3078-3094, Nov. 2024.
		\bibitem{brinton_digital}
		S. Wang, M. Chen,
		C. Shen, C. Yin, and C. G. Brinton, "Digital over-the-air federated learning in
		multi-antenna systems," \emph{IEEE Trans. Wireless Commun.}, vol. 23, no. 10, pp. 15125-15141, Oct. 2024. 
		\bibitem{kaibin_digital}
		J. Liu, Y. Gong, and K. Huang, "Digital over-the-air computation: Achieving high
		reliability via bit-slicing," \emph{IEEE Trans. Wireless Commun.}, vol. 24, no. 5, pp.  4101-4114, May 2025.
		\bibitem{huang_digital}
		Z. Wang, J. Yao, W. Xu, W. Shi, and K. Huang, "Digitalizing over-the-air computation via the novel
		complement coded modulation," \emph{IEEE Commun. Lett.}, vol. 30, pp. 812-816, 2026.
		\bibitem{zey_digital}
		Z. Li, C. Chen, and C. Fischione, "Channel-aware constellation design for digital
		OTA computation," arXiv: 2501.14675
		\bibitem{saeed_vector}
		S. Razavikia, J. M. B. Da Silva, and C. Fischione, "VecComp: Vector computing via MIMO digital over-the-air computation," \emph{IEEE Trans. Wireless Commun.}, vol. 25, pp. 16171-16186, May 2026.
		
		
		\bibitem{book}
		K. L. Du and M. N. S. Swamy, \emph{Wireless Communication Systems: From Fundamentals to 6G}. Cambridge Univ. Press, 2026.		
		\bibitem{fedcpu}
		S. M. Azimi-Abarghouyi and L. R. Varshney, "Compute-update federated learning: A lattice coding approach," \emph{IEEE Trans. Signal Process.}, vol. 72,  pp. 5213-5227, 2024.
		
		\bibitem{fedcpu_conf}
		S. M. Azimi-Abarghouyi and L. R. Varshney, "Federated learning via lattice joint source-channel coding," \emph{IEEE ISIT}, Athens, Greece, July 2024.
		
		\bibitem{rzamir}
		R. Zamir and M. Feder, "On universal quantization by randomized uniform/
		lattice quantizers," \emph{IEEE Trans. Inf. Theory}, vol. 38, no. 2, pp. 428-436, Mar. 1992.
		
		\bibitem{eldar}
		N. Shlezinger, M. Chen, Y. C. Eldar, H. V. Poor, and S. Cui, "UVeQFed: Universal vector quantization for
		federated learning," \emph{IEEE Trans. Signal Proc.}, vol. 69, pp. 500-514, 2021.
		
		\bibitem{nazer}
		B. Nazer and M. Gastpar, "Compute-and-forward: Harnessing interference through structured codes," \emph{IEEE Trans. Inf. Theory}, vol. 57,
		no. 10, pp. 6463-6486, 2011.
		
		\bibitem{nazermimo}
		J. Zhan, B. Nazer, M. Gastpar, and U. Erez, "MIMO compute-and-forward," \emph{IEEE Int. Symp. Inf. Theory (ISIT)}, Seoul, Korea, July 2009.
		
		\bibitem{azimi_cf1}
		S. M. Azimi-Abarghouyi, M. Hejazi, B. Makki, M. Nasiri-Kenari, and T. Svensson, "Decentralized compute-and-forward for ad hoc networks," \emph{IEEE Wireless Commun. Lett.}, vol. 5, no. 6, pp. 652-655, Dec. 2016.
		
		\bibitem{nazer_suc}
		B. Nazer, "Successive compute-and-forward," \emph{International Zurich Seminar on Communications (IZS)}, 2012.
		
		
		\bibitem{azimi_cf2} 
		M. Hejazi, S. M. Azimi-Abarghouyi, B. Makki, M. Nasiri-Kenari, and T. Svensson, "Robust successive compute-and-forward over multi-user multi-relay networks," \emph{IEEE Trans. Veh. Technol.}, vol. 65, no. 10, pp. 8112-8129, Oct. 2016.
		
		
		\bibitem{aluini}
		P. K. Vitthaladevuni and M. S. Alouini, "A recursive algorithm for the exact BER computation of
		generalized hierarchical QAM constellations," \emph{IEEE Trans. Inf. Theory}, vol. 49, no. 1, pp. 297-307, 2003.
		\bibitem{saeed}
		S. Razavikia, M. Kazemi, D. Gunduz, and C. Fischione, "Function computation over multiple access channels
		via hierarchical constellations," \emph{IEEE ISIT}, Guangzhou, China, July 2026.
		
		
		
	\end{thebibliography}
\end{document}